\documentclass[conference, 10pt, onecolumn]{ieeeconf}

\IEEEoverridecommandlockouts
\usepackage[usenames]{color}
\usepackage{enumerate}
\usepackage{url}
\usepackage{subfigure}
\usepackage{amsfonts,mathrsfs}
\usepackage{amssymb,amsmath}
\usepackage{verbatim}
\usepackage{acronym}
\usepackage{mathtools}
\usepackage{graphicx}
\usepackage{cite}
\usepackage{algorithm}
\usepackage[noend]{algpseudocode}
\usepackage{cite}

\def\fskip#1{}

\newtheorem{theorem}{Theorem}

\newtheorem{assumption}{Assumption}

\newtheorem{corollary}{Corollary}

\newtheorem{definition}{Definition}
\newtheorem{example}{Example}

\newtheorem{lemma}{Lemma}

\newtheorem{remark}{Remark}

\renewcommand{\theenumi}{\arabic{enumi}}

\def\1{{\bf 1}}

\newcommand{\remove}[1]{}

\def\la{\langle}
\def\ra{\rangle}
\def\argmin{\mathop{\rm argmin}}
\def\argmax{\mathop{\rm argmax}}

\begin{document}

\title{{\LARGE Online Assortment and Market Segmentation under Bertrand Competition Game with Set-Dependent Revenue Functions}\vspace{-0.3cm}}

\author{\authorblockN{S. Rasoul Etesami* \vspace{-0.3cm}}
\thanks{*Department of Industrial and Systems Engineering \& Coordinated Science Lab, University of Illinois at Urbana-Champaign, IL (etesami1@illinois.edu). This work is supported by NSF CAREER Award under Grant No. EPCN-1944403. An earlier version of this paper has appeared in CDC 2020 \cite{etesami2020dynamic}.}
}
\maketitle
\begin{abstract}
We consider an online assortment problem with $[n]:=\{1,2,\ldots,n\}$ sellers, each holding exactly one item $i\in[n]$ with initial inventory $c_i\in \mathbb{Z}_+$, and a sequence of homogeneous buyers arriving over a finite time horizon $t=1,2,\ldots,m$. There is an online platform whose goal is to offer a subset $S_t\subseteq [n]$ of sellers to the arriving buyer at time $t$ to maximize the expected revenue derived over the entire horizon while respecting the inventory constraints. Given an assortment $S_t$ at time $t$, it is assumed that the buyer will select an item from $S_t$ based on the well-known multinomial logit model, a well-justified choice model from the economic literature. In this model, the revenue obtained from selling an item $i$ at a given time $t$ critically depends on the assortment $S_t$ offered at that time and is given by the Nash equilibrium of a Bertrand game among the sellers in $S_t$. This imposes a strong dependence/externality among the offered assortments, sellers' revenues, and inventory levels. Despite that challenge, we devise a constant competitive algorithm for the online assortment problem with homogeneous buyers. It answers a question in \cite{zheng2019optimal-a} that considered the static version of the assortment problem with only one buyer and no inventory constraints. We also show that the online assortment problem with heterogeneous buyers does not admit a constant competitive algorithm. To compensate that issue, we then consider the assortment problem under an offline setting with heterogeneous buyers. Under a mild market consistency assumption, we show that the generalized Bertrand game admits a pure Nash equilibrium over general buyer-seller bipartite graphs. Finally, we develop an $O(\ln m)$-approximation algorithm for optimal market segmentation of the generalized Bertrand game which allows the platform to derive higher revenues by partitioning the market into smaller pools.       
\end{abstract}

\medskip
\section{Introduction}
Because of the rapid growth of the Internet and online markets such as Amazon, eBay, Airbnb, and Uber, the influence of online platforms on our daily decision-making is inevitable. In addition, online platforms play a major role in revenue management and supply-demand based on their pricing mechanisms. For instance, when a customer searches for a specific product (e.g., a TV) in an online marketplace (e.g., Amazon.com), the platform decides on what list of products to display to the customer. Similarly, when one is booking a room through Hotels.com or Orbitz.com, the online platform offers a list of available rooms at different hotels based on the searched data. As a result, depending on what list of items is offered, the online platform derives different revenues from potential buyers while meeting their demands using available items. In fact, based on a global survey \cite{evans2016rise}, the total market value of online platforms is growing rapidly and now exceeds a net value of 4.3 trillion U.S. dollars worldwide.   

In general, one can identify two types of online platforms in terms of their operation flexibility \cite{zheng2019optimal}. The first, known as the \emph{full control model}, allows the platform to have full control over offered assortments, prices, and the underlying supply-demand matching algorithms. As an example of the full control model, consider Uber, which not only offers rides to passengers, but also sets prices based on a revenue-maximizing optimization algorithm. However, unlike traditional firms, most online platforms do not dictate specific transaction prices. Instead, the platforms only create a trading environment and let the prices and revenues be determined organically as a result of interactions between buyers and sellers. This brings us to the \emph{discriminatory control model}, in which the platform sets the stage (e.g., decides what items to display), and the prices or potential matches are determined endogenously based on sellers'/buyers' choices. For instance, Airbnb only displays a list of available rooms to the customers; the rental prices are determined by competition among the property owners. As another example, eBay determines what products to display, and the prices/revenues are determined based on customers' demands and the competition among retailers.

Motivated by a variety of such examples, in this paper, we analyze a discriminatory control model for the marketplace problem under both online and offline settings. In the online setting, we consider a model in which each seller owns an indivisible item with a \emph{known quality} and \emph{private inventory}. As each item is owned by exactly one seller, we often refer to ``items" and ``sellers" interchangeably. There is a platform that can decide on what set of sellers to display to each upcoming homogeneous buyer. That induces a competition among the displayed sellers to set their items' prices, where the optimal prices are given by the Nash equilibrium of the Bertrand game induced by the set of displayed sellers. For instance, the Bertrand game can model the situation in which the hosts on Airbnb compete for potential guests by setting prices for their properties. Given the equilibrium prices, the guest either selects a room based on the multinomial logit (MNL) model or decides to leave the market without booking any room. Therefore, we are interested in devising an efficiently computable online policy for the platform to maximize the aggregate revenue of the sellers over the entire horizon. In particular, we want our online algorithm to be competitive with respect to the best in-hindsight policy that knows the total number of buyers and the inventory levels a priori.

Unfortunately, we will show that for \emph{heterogeneous} buyers with different evaluations of items' qualities, there is no constant competitive online algorithm for the platform. To address that issue, we also consider an offline marketplace with heterogeneous buyers in which all the buyers are available in the market, i.e., they do not joining sequentially. As a result, the equilibrium prices of the sellers are now given by a so-called \emph{generalized Bertrand game} that is induced by a two-sided market in which buyers and sellers are on different sides. Here we are interested in the existence of a pure Nash equilibrium in the generalized Bertrand game that extends the existence result of \cite{zheng2019optimal} from the case of a single buyer to the multi-buyer case over general bipartite graphs. In particular, we consider the problem of optimal market segmentation over the generalized Bertrand game, where the goal is to segment the entire market into smaller pools such that the aggregate revenue obtained at the equilibrium of the submarkets improves on the revenue obtained at the equilibrium of the original market. It is worth noting that market segmentation is a practical approach for many online platforms such as Uber in which the platform decides what buyers and sellers are visible to each other so as to improve the total derived revenue \cite{banerjee2017segmenting}.

\subsection{Related Work}
The design of an online marketplace with different objectives has been studied in the past literature \cite{lin2017networked, banerjee2017segmenting,birge2019optimal,mcfadden1986choice}. However, unlike prior work \cite{birge2019optimal,banerjee2017segmenting} that uses monotone distribution functions to capture supply and demand curves, we use the well-known multinomial logit (MNL) model \cite{mcfadden1974conditional,du2016optimal,feldman2018capacitated} to capture buyers' choice behavior. That, in turn, allows us to describe buyers' demands in terms of purchase probabilities for substitutable items. Through use of such demand functions, the competition among displayed sellers is captured using the Bertrand game, which has been extensively studied in the economic literature \cite{vives1985efficiency,vives1999oligopoly} for modeling oligopolistic competition in actual markets. Perhaps our work is most related to \cite{zheng2019optimal-a} and \cite{zheng2019optimal}, in which the static version of the problem that we consider here was studied, i.e., when there is only one buyer with no inventory constraints. It was shown in \cite{zheng2019optimal-a} that the revenue function for a single-buyer, one-shot market has a nice quasi-convex structure that allows one to find the optimal assortment efficiently. In particular, the authors raised the question of designing an optimal dynamic mechanism whereby the buyers arrive and depart, and the sellers have limited capacity for products. In this paper, we answer that question by devising the first constant-factor competitive algorithm that can be implemented efficiently in real-time.

This work is also closely related to dynamic online assortment \cite{golrezaei2014real,fata2019multi,wang2017consumer,gong2019online,feng2019linear,rusmevichientong2010dynamic} and dynamic revenue management \cite{talluri2004revenue,gallego2015online,dong2009dynamic,goldberg2016asymptotic,chen2004coordinating}, whose goal, broadly speaking, is to dynamically choose assortments/prices in order to maximize the aggregate revenue obtained by selling items subject to inventory constraints. Such problems have been extensively studied under various settings, such as the stochastic demand arrival model \cite{talluri2004revenue,goldberg2016asymptotic,chen2004coordinating}, the adversarial demand model in which the sequence of buyers' types can be chosen adversarially \cite{golrezaei2014real}, and reusable items, i.e., an allocated product can be returned to the inventory after some random amount of time \cite{gong2019online,feng2019linear}. We refer to \cite{chen2012pricing} for a comprehensive survey of recent advances on dynamic pricing and inventory management problems.

On the other hand, except for a handful of results \cite{gallego2014dynamic,heese2018effects,wang2017consumer}, the effect of competition or externalities in dynamic assortment problems has not been addressed before. In fact, almost all the earlier results on dynamic online assortment (see, e.g., \cite{fata2019multi,talluri2004revenue,golrezaei2014real}) consider fixed-revenue items, meaning that the revenue obtained from selling an item is fixed and does not depend on the offered assortments. That makes linear programming (LP) methods quite amenable to use in design of  competitive online algorithms, and such LP methods have been used in several earlier research efforts \cite{feng2019linear,golrezaei2014real}. Unfortunately, the dependence of revenues on the offered assortments creates extra externalities that make the use of conventional methods, such as dual-fitting or dynamic programming formulations, more complex. To handle such externalities, in this work we take a different approach and use a novel charging argument that allows us to compare the revenue obtained from our online algorithm with that of the optimal clairvoyant online benchmark.       

It is worth mentioning that our work is also related to online learning for LPs with packing constraints \cite{badanidiyuru2013bandits,devanur2011near,gupta2016experts}. The reason is that one can use an LP with packing constraints to upper-bound the revenue of the optimal clairvoyant online algorithm. As a result, generating a feasible solution in an online fashion to such an LP whose objective value is competitive with the offline LP's will automatically deliver a competitive online algorithm for our problem. A typical approach is to dynamically learn/estimate the dual variables of the offline LP by using some black-box online learning algorithms (e.g., the \emph{multiplicative weight} update rule \cite{gupta2016experts}) and use the estimated dual variables to guide the primal solution generated by the algorithm. However, the LP formulation in our setting has exponentially many variables that are associated with all the feasible assortments. Therefore, it is not clear, without extra effort, how this approach can be used to devise a polynomial-time competitive algorithm for our online assortment problem.

Finally, the generalized Bertrand game that we consider in this paper is an extension of the single-buyer Bertrand game given in \cite{zheng2019optimal}. To the best of our knowledge, Bertrand games over general bipartite graphs and MNL demands have not been addressed in the past literature. Here we should mention that other oligopolistic competitions, such as Cournot competition over general bipartite graphs, have been studied in \cite{bimpikis2019cournot,abolhassani2014network}. However, a major challenge in analyzing the generalized Bertrand game compared to the network Cournot game is the restriction of the players' strategy spaces. More precisely, in the network Cournot game, a player can control its supply to each of its buyers, while in the generalized Bertrand game, a player can only set one price, and the demands are computed endogenously. Moreover, establishing the existence of a pure Nash equilibrium in network Cournot games is typically done by reduction to potential games \cite{monderer1996potential} or submodular games \cite{topkis1979equilibrium}, a property that does not seem to hold in our generalized Bertrand game.

\subsection{Contributions and Organization}
In Section \ref{sec:formulation}, we formally define the online assortment problem, which extends the existing literature from the single-buyer static model to a dynamic setting with inventory constraints. In Section \ref{eq:preliminary}, we show that the online assortment problem with heterogeneous buyers does not admit a constant competitive algorithm. Subsequently, we restrict our attention to the online assortment problem with identical buyers. We then provide a linear programming upper bound for the optimal revenue of any online algorithm and establish some results on the substitutability of the revenue functions. In Section \ref{sec:main}, we devise a constant competitive algorithm for the online assortment problem with homogeneous buyers. The proposed algorithm is a hybrid deterministic algorithm whose analysis is based on a novel charging argument. The approach extends the existing results on online assortment problems with fixed revenues to set-dependent revenue functions. We also provide some numerical results to demonstrate the efficiency of the proposed algorithm beyond its theoretical guarantee. 

Because of the nonexistence of competitive online algorithms with heterogeneous buyers and in order to extend our results to heterogeneous buyers, in Section \ref{sec-offline}, we introduce an offline generalized Bertrand game over general bipartite graphs. In the generalized game, buyers can have different evaluations of items' qualities, and sellers have a limited number of items. Despite the highly nonlinear structure of the sellers' payoff functions, we show that under a mild market consistency assumption, the generalized Bertrand game admits a pure Nash equilibrium. We also provide an approximation algorithm for the optimal segmentation of the generalized Bertrand game, which allows the offline platform to derive higher revenues by partitioning the market into smaller pools. We conclude the paper in Section \ref{sec:conclusion}, and relegate omitted proofs and numerical results to Appendix I and Appendix II, respectively.

\section{Problem Formulation}\label{sec:formulation}
        
In this section, we introduce the online assortment problem. We will postpone the offline setting to Section \ref{sec-offline}. We consider an online assortment problem with a set $[n]=\{1,\ldots,n\}$ of sellers, where the $i$th seller holds item $i\in[n]$ with an initial inventory $c_i\in \mathbb{Z}_+$. Here $c_i$ and, more generally, the inventory state of seller $i$ is private information that is known only to seller $i$ and not to any other seller or the platform. As noted earlier, as each item is assigned to exactly one seller, throughout this paper we use the terms \emph{sellers} and \emph{items} interchangeably. Each item $i\in [n]$ has a fixed quality $\theta_i\in \mathbb{R}$, which is public information and known to all the buyers, sellers, and the platform.\footnote{The fact that items' qualities are public information is not a strong assumption and holds in many practical settings. For instance, the quality of hotel rooms is specified by an integer number of stars, or the quality of products is available based on customers' reviews on a scale of $0$ to $5$. More generally, there are various methods to quantify service qualities, such as SERVQUAL, which is a multidimensional scale developed to assess customer perceptions of service quality in retail businesses.} Without loss of generality, we assume that the items are sorted according to their qualities, i.e., $\theta_1\ge \ldots\ge \theta _n$. We also define a ``no-purchase" item $\{0\}$ with $\theta_0=0$ to represent the case in which a buyer decides to leave the market without making a purchase.  We should mention that unlike \cite{zheng2019optimal}, which only considers nonnegative qualities, we allow negative qualities. Such an extension creates a more realistic market in which a negative-quality item captures the case in which a buyer strongly prefers to leave the market rather than purchase that item.

\subsection{Sellers' Equilibrium Prices, Demands, and Revenues}
We consider a sequence of \emph{homogeneous} buyers that arrive at discrete time instances $t=1,2,\ldots,m$, for some unknown integer $m\in \mathbb{Z}_+$. As in \cite{zheng2019optimal} and \cite{mcfadden1986choice}, we adopt the random utility model for the purchase probability of the buyers. In the random utility model, the buyer has a private random preference $\zeta_i$ about the $i$th item. Given a vector of nonnegative prices $\bar{\boldsymbol{p}}=(\bar{p}_1,\ldots,\bar{p}_n)$ on the items,\footnote{By convention, we assume that the price of no-purchase item $\{0\}$ is normalized to $p_0=0$.} the buyer derives utility $u_i=\zeta_i+\theta_i-\bar{p}_i$ from purchasing item $i\in [n]$.  Therefore, the buyer purchases item $i$ with the highest utility, i.e., $i=\argmax_{j\in [n]\cup\{0\}}u_j$, where ties are broken arbitrarily. In particular, under the assumption that $\zeta_i$ are i.i.d. random variables with a Gumbel distribution, we obtain the well-known multinomial logit (MNL) purchase probabilities \cite{mcfadden1974conditional,du2016optimal,feldman2018capacitated}. Thus, for every $i\in [n]\cup\{0\}$ we have
\begin{align}\label{eq:MNL-probability}
\bar{q}_i(\bar{\boldsymbol{p}}):=\mathbb{P}(i=\argmax_{j\in [n]\cup\{0\}}u_j)=\frac{e^{\theta_i-\bar{p}_i}}{\sum_{j\in [n]\cup\{0\}}e^{\theta_j-\bar{p}_j}}.
\end{align} 
The probability $\bar{q}_i(\bar{\boldsymbol{p}})$ can also be viewed as the \emph{expected demand} of a buyer for item $i$ given posted prices $\bar{\boldsymbol{p}}$.  

In this paper, we consider a discriminatory control model wherein, upon the arrival of a buyer, the online platform can decide on what subset of items to display to the buyer, but has no control over the posted prices, which are determined endogenously by competition among the displayed sellers. To capture the competition among the sellers, we use a Bertrand competition game \cite{zheng2019optimal,vives1999oligopoly,vives1985efficiency}, wherein the seller of each item sets a price to maximize its own revenue. More precisely, let $S\subseteq [n]$ be the set of items that are displayed by the online platform to a specific buyer. Consider a Bertrand game between the sellers (players) in $S$, where an action for seller $i$ is to set a price $\bar{p}_i\ge 0$ on its item. Given the profile of prices $\bar{\boldsymbol{p}}(S)=(\bar{p}_i, i\in S)$, the revenue of seller $i\in S$ is given by the expected amount of item $i$ that is sold, based on \eqref{eq:MNL-probability}, multiplied by the price of that item, i.e., $R_i(\bar{\boldsymbol{p}}(S))=\bar{p}_i(S) \times \bar{q}_i(\bar{\boldsymbol{p}}(S))$. Assuming that in the induced Bertrand game $(S, \bar{\boldsymbol{p}}(S), \boldsymbol{R}(\bar{\boldsymbol{p}}(S)))$, each seller maximizes its own revenue, the \emph{unique} pure Nash equilibrium prices can be computed in a closed form as \cite[Theorems 1 \& 2]{zheng2019optimal-a}:
\begin{align}\label{eq:p-q}
p_i(S)=\frac{1}{1-q_i(S)} \ \ \ \ i\in S,
\end{align}     
where $q_i(S)$ is the MNL demand probability \eqref{eq:MNL-probability} computed at the equilibrium prices $\boldsymbol{p}(S)$. If we substitute the probabilities from \eqref{eq:MNL-probability} into \eqref{eq:p-q} and solve for equilibrium prices, we obtain a closed-form solution for the equilibrium prices and the equilibrium demands in terms of the ``no-purchase" probability $q_0(S)$ as \cite[Theorem 2]{zheng2019optimal-a}:
\begin{align}\label{eq:p-q_0}
q_i(S)=V\big(q_0(S) e^{\theta_i-1}\big), \  \ \ \  i\in S.
\end{align}  
Here, $V(x):\mathbb{R}_+\to [0,1)$ is a strictly increasing and concave function given by the unique solution of $y\exp({\frac{y}{1-y}})=x$, and $q_0(S)$ is given by the unique solution of the equation 
\begin{align}\label{eq:q_0}
\sum_{i\in S}V(q_0(S)e^{\theta_i-1})=1-q_0(S). 
\end{align}
Thus, by offering assortment $S$, the expected revenue that seller $i$ derives at the Nash equilibrium equals to
\begin{align}\label{eq:rev-q}
R_i(S)=\begin{cases}\frac{q_i(S)}{1-q_i(S)}, &\mbox{if}  \ \ \  i\in S,\\
0&\mbox{if}  \ \ \  i\notin S,
\end{cases}
\end{align} 
where $q_i(S)$ is given by \eqref{eq:p-q_0}. Finally, we let
\begin{align}\nonumber
R(S)=\sum_{i\in [n]}R_i(S)=\sum_{i\in S}\frac{q_i(S)}{1-q_i(S)}
\end{align}
be the expected revenue obtained by offering assortment $S$. 

\begin{example}
A closer look at \eqref{eq:rev-q} reveals that the revenue for an item depends on the offered assortment. As we shall see, that makes the analysis of assortment dynamics much more complicated than the conventional models in which items have \emph{fixed} revenues, regardless of the offered assortments. Such dependency creates externalities between items, as the revenue of an item now depends on other items that are bundled with it. To illustrate such dependency, let us consider an example with two items of qualities $\theta_1=1, \theta_2=2$. If the platform offers assortment $S=\{1\}$, which includes the first item only, then by solving \eqref{eq:q_0} we get $q_0(\{1\})=0.64$, and thus $q_1(\{1\})=1-q_0(\{1\})=0.36$. In particular, the equilibrium price for item $1$ in this case is given by $p_1(\{1\})=\frac{1}{1-0.36}=1.56$. Therefore, from \eqref{eq:rev-q}, the expected revenue of item $1$ is equal to $R_1(\{1\})=0.36\times 1.56=0.56$. As there is only one item in the assortment, that value is also equal to the total expected revenue, i.e., $R(\{1\})=0.56$. Similarly, if the platform offers assortment $S=\{2\}$,  then by solving \eqref{eq:q_0} we get $q_0(\{2\})=0.5$,  and thus $q_2(\{2\})=0.5$, $p_2(\{2\})=2, R_2(\{2\})=R(\{2\})=1$.
  
Now consider the case that in which platform offers both items, i.e., $S=\{1,2\}$. Then, by solving \eqref{eq:q_0}, we get $q_0(\{1,2\})=0.34$. Using \eqref{eq:p-q_0} we have $q_1(\{1,2\})=0.23$ and $q_2(\{1,2\})=0.43$. Therefore, in this case, the equilibrium price for item $1$ equals $p_1(\{1,2\})=\frac{1}{1-0.23}=1.29$, which is lower than $p_1(\{1\})=1.56$ in the previous case. The reason is that offering a larger assortment increases the competition among the sellers so that they choose to set their prices lower in order to attract more demand. In particular,  by offering $S=\{1,2\}$, the expected revenues for items $1$ and $2$ are $R_1(\{1,2\})=0.29$ and $R_2(\{1,2\})=0.75$, respectively,  which are less than their corresponding values $R_1(\{1\})=0.56$ and $R_2(\{2\})=1$. Note, however, that the total revenue of offering assortment $S=\{1,2\}$ is $R(\{1,2\})=1.04$, which is much greater than $R(\{1\})=0.56$.  The reason is that item $2$ is a better item whose addition to the assortment can extract more revenue. Similarly, $R(\{1,2\})=1.04$ is also greater than $R(\{2\})=1$, because including more items in the assortment creates positive externality in terms of attracting higher demand.
\end{example}

As we mentioned earlier, each seller has full information about its inventory level at any time.  However, a seller does not know anything about the inventory state of other sellers as it is private information that sellers do not wish to share with their competitors.  In particular, a seller cannot use the information on others' inventories to compute its equilibrium price. Note that given homogeneous buyers, there is no incentive for a risk-averse seller $i$ to set its price differently than its unit-inventory equilibrium price. The reason is that upon arrival of a \emph{unit-demand} buyer,  if seller $i$ assumes that others have higher inventories (and so have the incentive to offer their items at a lower price), he will also offer his item at a lower price to be competitive with others. However, that reduces the expected stage revenue of seller $i$.  As the revenue of seller $i$ is additive over the sequence of buyers and the buyers are homogeneous,  seller $i$ can focus on maximizing its expected revenue at any single stage, assuming that others choose their unit-inventory equilibrium prices. Hence,  by definition of Nash equilibrium, the revenue-maximizing price for seller $i$ at any stage is the unit-inventory equilibrium price. On the other hand, the reason why the platform should design its algorithm based on the equilibrium prices of sellers is that the Bertrand game induces a potential game \cite{monderer1996potential}. Consequently, if each seller sets its price most naturally by best responding to others' posted prices (i.e., setting to $\argmax_{p_i} R_i(\boldsymbol{p}(S))$), then the prices converge quickly to the equilibrium of the Bertrand game through the repetition of interactions. Thus, equilibrium prices can be viewed as stationary prices in the steady-state game.\footnote{For instance, in the case of Airbnb, the platform has access to prices at the time when it chooses an assortment \cite{airbab}. That leads to a repeated game in which Airbnb chooses an assortment based on current prices, and that, over time, induces a change in prices as the hosts optimize.}

\begin{remark}\label{rem:potential}
For an assortment $S$ and price profile $\boldsymbol{p}=\boldsymbol{p}(S)$, the function 
\begin{align}\nonumber
\Phi(\boldsymbol{p})=\frac{\prod_{j\in S}p_je^{\theta_j-p_j}}{\sum_{j\in S\cup\{0\}}e^{\theta_j-p_j}},
\end{align}
serves as a potential function for the single-buyer Bertrand game. This function is simply the (normalized) sum of the sellers' revenues in the logarithm space. More precisely, given price profile $\boldsymbol{p}=(p_i,p_{-i})$, we have $\ln \Phi(\boldsymbol{p}) \propto \sum_{j\in S}\ln r_j(\boldsymbol{p})$, where $r_j(\boldsymbol{p})=p_jq_j$ denotes the revenue of seller $j$ for posted price $p_j$, and the corresponding demand $q_j$ that is calculated according to \eqref{eq:MNL-probability}. In particular,  for any other price profile $\boldsymbol{p}'=(p'_i,p_{-i})$, the potential difference $\ln \Phi(\boldsymbol{p})-\ln \Phi(\boldsymbol{p}')$ equals
\begin{align}\nonumber
\ln \Phi(\boldsymbol{p})\!-\!\ln \Phi(\boldsymbol{p}')\!=\!\ln\!\Big(\frac{p_ie^{\theta_i-p_i}}{p'_ie^{\theta_i-p'_i}}\!\Big)\!+\!\ln\!\Big(\frac{e^{\theta_i-p'_i}\!+\!\sum_{j\in S\cup\{0\}\!\setminus\!\{i\}}e^{\theta_j-p_j}}{e^{\theta_i-p_i}\!+\!\sum_{j\in S\cup\{0\}\!\setminus\!\{i\}}e^{\theta_j-p_j}}\!\Big)\!=\!\ln r_i(\boldsymbol{p})\!-\!\ln r_i(\boldsymbol{p}').
\end{align}
That shows that any seller's best response increases the potential function, and the prices will converge to a pure Nash equilibrium \cite{monderer1996potential,zheng2019optimal}.  Thus, sellers do not need to be very intelligent to compute their equilibrium prices; simply best responding to others' prices will guarantee convergence to equilibrium prices.
\end{remark}

The platform can use different signaling mechanisms to communicate its assortment decision to the sellers. That allows sellers to be aware of their competitors and optimize their equilibrium prices.  For instance, upon the arrival of a new buyer, the platform can choose an assortment, announce it to the selected sellers, and request their prices.  Alternatively, once a buyer searches for an item, the platform can display selected items to everyone so the sellers automatically will be aware of whether they are selected into the offered assortment. Of course, the platform can itself perform the equilibrium price computation in a centralized manner (as it has access to all the necessary information, such as items' qualities and the selected assortment).  However, the advantage of using the signaling mechanism is that it allows the equilibrium prices to be computed in a decentralized manner by the sellers (Remark \ref{rem:potential}), hence substantially reducing the computational load of the platform.

\subsection{Online Optimization for the Platform}
The online assortment problem that is faced by the online platform is that of selecting a sequence of assortments subject to inventory constraints so as to maximize the aggregate expected revenue
\begin{align}\label{eq:rev-objective}
\sum_{t=1}^{m}\mathbb{E}[R(S_t)]=\sum_{t=1}^{m}\mathbb{E}\Big[\sum_{i\in S_t}\frac{q_i(S_t)}{1-q_i(S_t)}\Big],
\end{align} 
where $S_1,S_2,\ldots,S_m$ are random variables denoting the assortments offered by the platform at times $t=1,2,\ldots,m$. Here, the offered assortments must satisfy the inventory constraints, meaning that an item $i\in [n]$ can be included in $S_t$ at time $t$ only if it is still available at that time, i.e., if its $c_i$ units have not been fully sold to buyers $1,2,\ldots,t-1$. It is worth noting that the available inventory of an item at a time $t$ depends not only on the assortments offered by the platform up to that time, but also on the buyers' choice realizations up to time $t$. Therefore, an online algorithm must satisfy the inventory constraints for any realized sample path of buyers' choices. 

The platform can use either a deterministic or a randomized online algorithm, where by an \emph{online} algorithm we refer to an algorithm that does not know the total number of buyers $m$ or the initial inventory of the sellers $\boldsymbol{c}$, nor does it have access to the random choices realized by buyers. a priori.  In other words, the only information to which an online algorithm has access before it makes a decision at time $t$ is the quality of items $\boldsymbol{\theta}$ and whether or not an item is available at that time. We compare the performance guarantee of our online algorithm against the best \emph{clairvoyant online algorithm} that knows the item's qualities $\boldsymbol{\theta}$, the total number of buyers $m$, and the current and past state of inventories but not the future realization of the choices made by the buyers.  More precisely, for each future buyer and an assortment, the clairvoyant algorithm knows the probability that the buyer will purchase an item from that assortment but does not know the exact choice that the buyer will make.  In other words, the clairvoyant online algorithm knows the choice model but does not know the realization of the random choices of a future customer. To evaluate the performance of an online algorithm, we use the notion of a competitive ratio, as defined next.

\begin{definition}
Given $\alpha>0$, an online algorithm for the dynamic assortment problem is called $\alpha$-\emph{competitive} if for every instance it achieves in expectation an $\alpha$-fraction of the total revenue obtained by any clairvoyant online algorithm.
\end{definition}

\section{Preliminary Results}\label{eq:preliminary}

In this section, we prove several useful lemmas. To devise a competitive algorithm, we first derive an upper bound for the expected revenue that any clairvoyant online algorithm can obtain and use it as a benchmark to compare the performance of our online algorithm with that upper bound. This upper bound can be obtained by writing an offline LP whose optimal objective value is no less than the expected revenue of any feasible clairvoyant online algorithm. (We refer the reader to \cite{fata2019multi,golrezaei2014real,feng2019linear} for a similar method for devising online competitive algorithms.) Let us assume that the number of buyers $m$ and inventory levels $c_i, \forall i$ are known and consider an offline LP whose optimal objective value (OPT) is given by
\begin{align}\label{eq:LP-Game}
\mbox{OPT}=\max &\ \ \sum_{t=1}^m\sum_{S} R(S)y^t(S)\cr 
\mbox{s.t.}  &\ \ \sum_{t=1}^m\sum_{S} \boldsymbol{q}(S)y^t(S)\leq \boldsymbol{c},\cr 
& \ \ \boldsymbol{y}^t\in \Delta_{2^n}, \forall t=1,\ldots,m.
\end{align} 
Here, $\boldsymbol{c}=(c_1,\ldots,c_n)'$ is the vector of initial inventories, and $\boldsymbol{q}(S)=(q_1(S),\ldots,q_n(S))'$ is the column vector of equilibrium purchasing probabilities, where $q_i(S)$ is obtained from \eqref{eq:p-q_0} and is the equilibrium demand for item $i$ given that assortment $S$ is offered. Moreover, for each $t$, the variable vector $\boldsymbol{y}^t=(y^t(S): S\subseteq [n])$ belongs to the probability simplex $\Delta_{2^n}=\{\boldsymbol{y}\ge 0: \sum_{S\subseteq [n]}y(S)=1\}$, where $y^t(S)$ can be viewed as the probability that the algorithm will offer assortment $S$ to the buyer at time $t$. Thus, the first $n$ constraints in \eqref{eq:LP-Game} capture the inventory constraints in a vector form.

Next, we show that any clairvoyant online algorithm generates a feasible solution to the LP in \eqref{eq:LP-Game} with an objective value that is equal to the expected revenue of that clairvoyant algorithm.  That shows that the OPT provides an upper bound for the revenue obtained by any clairvoyant online algorithm. To see that,  for any $t=1,\ldots,m$, let $S_t$ be a random variable denoting the assortment offered by the clairvoyant online algorithm at time step $t$, and $\eta_t$ be a random variable denoting the item that is purchased by the buyer at time $t$.  Note that $S_t$ may depend on $S_1,\ldots,S_{t-1}$ and $\eta_1,\ldots,\eta_{t-1}$, such that at any time $t$, the clairvoyant online algorithm can observe the past decisions of the buyers up time $t-1$ (and hence knows the state of the inventories up to time $t-1$),  and then offers the assortment $S_t$. The only information that is not available to the clairvoyant online algorithm at the time of offering $S_t$ is the \emph{realization} of the future buyer's choices $\eta_t,\ldots,\eta_m$. Then, for every realized sample path $\sum_{t=1}^m\sum_S\boldsymbol{1}\{S_t=S, \eta_t=i\}\leq c_i$, $\forall i\in [n]$, where $\boldsymbol{1}\{\cdot\}$ is the indicator function. By taking expectation from this inequality over all sample paths, we can write
\begin{align}\nonumber
\sum_{t=1}^m\sum_S\mathbb{P}(S_t=S)q^t_i(S)&=\sum_{t=1}^m\sum_S\mathbb{P}(S_t=S)\mathbb{P}(\eta_t=i|S_t=S)\cr 
&=\sum_{t=1}^m\sum_S\mathbb{E}[\boldsymbol{1}\{S_t=S, \eta_t=i\}]\leq c_i. 
\end{align}
Thus, setting $y^t(S)=\mathbb{P}(S_t=S)$ forms a feasible solution to the LP in \eqref{eq:LP-Game} whose objective value is equal to the expected revenue of the online clairvoyant algorithm, that is, $\sum_{t=1}^m\sum_{S} R(S)\mathbb{P}(S_t=S)$.  That shows that the OPT provides an upper bound for the revenue obtained by any clairvoyant online algorithm.

We note that in the online assortment problem, we restrict our attention to identical buyers such that $R(S)$ and $\boldsymbol{q}(S)$ do not depend on $t$. The reason is that if $R^t(S), \boldsymbol{q}^t(S)$ can also depend on the type of buyers, then, as is shown in the following theorem, no online algorithm can achieve a constant competitive ratio.

\begin{theorem}
There is no constant competitive algorithm for the online assortment problem with heterogeneous buyers' evaluations.
\end{theorem}
\begin{proof}
Consider a single item with unit inventory $c=1$, and, for simplicity, we drop all the indices that denote that item. Consider an adversary who first selects the number of buyers $m$ and their quality evaluations (types) $\theta_t, t\in[m]$, and then reveals them sequentially over time to the seller. Let $M>0$ be a large number, and for $t\in \mathbb{Z}_+$, we choose $\theta_t$ such that the expected revenue of selling the item to buyer $t$ equals $M^t$. Note that such a selection is possible because the expected revenue function $\frac{q}{1-q}$ is a continuously increasing function and admits any value in $[0, \infty)$. Moreover, as $M$ is a large number, for any buyer $t$, the equality $\frac{q}{1-q}=M^t$ implies that $q\ge \frac{1}{2}$. In other words, if the item is offered to any buyer, it is purchased with probability at least $\frac{1}{2}$. 

Now let $\mathcal{I}=\big\{(\theta_1), (\theta_1,\theta_2), (\theta_1,\theta_2, \ldots, \theta_t), \ldots\big\}$ be the set of input sequences that can be chosen by the adversary. To derive a contradiction, let us assume that there exists an $\alpha$-competitive (randomized) algorithm ALG for some $\alpha>0$. Then, ALG must offer the item to the first arriving buyer with probability $\beta_1\ge \alpha$. The reason is that for the unit length instance $(\theta_1)$, that is, when the adversary chooses to send only one buyer of type $\theta_1$, the expected revenue obtained by ALG is $\beta_1 M$. On the other hand, the optimal clairvoyant algorithm that knows the adversary's sequence will offer the item to the first buyer with probability one and derives an expected revenue of $M$. As the expected revenue of ALG for \emph{any} instance must be at least $\alpha$-fraction of the optimal revenue, we have $\frac{\beta_1 M}{M}\ge \alpha$.

Using induction, we next show that for any $t\ge 2$, the ALG must offer the item to buyer $t$ with some probability $\beta_t\ge \alpha$. First, we note that after revealing the first $t-1$ buyers to the ALG, all the remaining instances in $\mathcal{I}$ are indistinguishable so that the ALG cannot draw any conclusion about the input sequence. Moreover, the ALG must perform well on any input sequence. Therefore, if the adversary chooses the input instance to be $(\theta_1,\ldots,\theta_t)$, the expected revenue of ALG is at most
\begin{align}\nonumber
&\sum_{\ell=1}^{t-1}M^{\ell}+\mathbb{P}\big\{\mbox{item is available at time $t$}\big\}\beta_t M^t\cr 
&\qquad\leq \sum_{\ell=1}^{t-1}M^{\ell}\!+\mathbb{P}\big\{\mbox{item is available at $t-1$}\big\}\big(1-\frac{\beta_{t-1}}{2}\big)\beta_t M^t\cr 
&\qquad\leq \sum_{\ell=1}^{t-1}M^{\ell}+\Big(\prod_{\ell=1}^{t-1}\big(1-\frac{\beta_{\ell}}{2}\big)\Big)\beta_t M^t\cr 
&\qquad< \frac{M^t}{M-1}+\big(1-\frac{\alpha}{2}\big)^{t-1}\beta_t M^t,
\end{align}  
where the first inequality holds because the probability that the item is available at time $t$ is equal to the probability that the item is available at time $t-1$ multiplied by the sum of two conditional probabilities: i) the probability that the item is not offered at time $t-1$ given that it is available (this probability is $1-\beta_{t-1}$), and ii) the probability that the item was offered at time $t-1$ and was not sold, given that it was available (this probability is at most $\frac{\beta_{t-1}}{2}$). The last inequality holds by the induction hypothesis as $\beta_{\ell}\ge \alpha, \forall \ell\leq t-1$. On the other hand, the expected revenue of the optimal clairvoyant algorithm in this instance is $M^t$, as it offers the item with a probability of one to buyer $t$. Thus,
\begin{align}\label{eq:M-alpha}
 \frac{M^t}{M-1}+\big(1-\frac{\alpha}{2}\big)^{t-1}\beta_t M^t> \alpha M^t.
\end{align}
Since $M$ is a large number (e.g., $M= \frac{2}{\alpha^2}+1$), and $t\ge 2$, we conclude that $\beta_t\ge \alpha$. This completes the induction. 

Finally, as $\beta_t\in [0, 1], \forall t$, using \eqref{eq:M-alpha}, we must have,
\begin{align}\nonumber
 \frac{\alpha-\frac{1}{M-1}}{\big(1-\frac{\alpha}{2}\big)^{t-1}}< \beta_t\leq 1, 
\end{align}
which cannot hold for a sufficiently large $t$. This contradiction shows that no online algorithm can be $\alpha$-competitive for any constant $\alpha>0$.
\end{proof}

\subsection{Substitutability Property and Items' Heaviness}
Here, we provide several useful lemmas for our later analysis. All the proofs of this section can be found in Appendix I. We start with the following so-called \emph{substitutability} property which essentially says that the demand for a particular item decreases as more items are offered.    

\begin{lemma}\label{lemm:subtitute}
For any assortment $S$ and two items $i\in S, j\notin S$, we have $q_i(S)\ge q_i(S\cup \{j\})$ and $R_i(S)\ge R_i(S\cup\{j\})$. 
\end{lemma}

\begin{definition}\label{def:heavy}
Let $\lambda\in (\frac{1}{2}, 1)$ be a fixed parameter. An item $i$ is called \emph{heavy} if offering it alone obtains at least $\lambda$-fraction of the market share, i.e., $q_i(\{i\})\ge \lambda$. An item $i$ is \emph{heavier} than item $j$ if $q_i(\{i\})\ge q_j(\{j\})$. We use $[h]=\{1,2,\ldots,h\}$ to denote the set of heavy items, if any, and refer to any item $i\notin[h]$ as a \emph{light} item.     
\end{definition}

Note that if it has access to items' qualities, an online algorithm can use \eqref{eq:p-q_0} and \eqref{eq:q_0} to easily compute the expected demands $q_i(\{i\})$ a priori. Therefore, for a fixed threshold $\lambda$, we may assume without loss of generality that an online algorithm knows the heaviness of all the items. Later, we will optimize the competitive ratio of our devised algorithm over $\lambda\in (0, \frac{1}{2})$ to obtain the optimal heaviness threshold.

\begin{remark}
As items are sorted according to their qualities $\theta_1\ge \ldots\ge \theta_n$, using  \eqref{eq:p-q_0} and the monotonicity of $V(\cdot)$, one can see that the same order must hold on the heaviness of the items, i.e., $q_1(\{1\})\ge\ldots\ge q_n(\{n\})$.
\end{remark}

\begin{lemma}\label{lemm-heavy-two}
Let $f(\lambda)=\max\{1+(\frac{1-\lambda}{\lambda})^2, \frac{1}{\lambda}\}$. Consider an arbitrary assortment $S$, and let $i\in [h]$ be a heavy item that is at least as heavy as any other item in $S$. Then, offering $i$ alone obtains at least $f(\lambda)$-fraction of the expected revenue of offering $S$, i.e., $R(S)\leq f(\lambda)\frac{q_i(\{i\})}{1-q_i(\{i\})}$.
\end{lemma}

The next lemma shows that in the absence of heavy items, a simple greedy algorithm that offers all the available items at each round is $\frac{1}{2}$-competitive. We shall use this lemma latter to connect the revenue obtained during the second phase of our algorithm to the revenue of the optimal clairvoyant benchmark.

\begin{lemma}\label{lemm:basic-Big}
An online algorithm that at each time offers all the available items achieves at least $\frac{1}{2}$ of optimal value of the following LP:
\begin{align}\nonumber
\max &\ \ \sum_{t=1}^m\sum_{S} \sum_i q_i(S)y^t(S)\cr 
\mbox{s.t.}  &\ \ \sum_{t=1}^m\sum_{S} \boldsymbol{q}(S)y^t(S)\leq \boldsymbol{c},\cr 
& \ \ \boldsymbol{y}^t\in \Delta_{2^n}, \forall t=1,\ldots,m.
\end{align}  
\end{lemma}

\medskip
\section{A Competitive Algorithm for the Online Assortment Problem}\label{sec:main}

In this section, we first describe a deterministic online assortment algorithm and then prove its performance guarantee. The proposed algorithm is very simple and requires only finding the heaviest item at each round with an overall computation of $O(mn)$. The algorithm consists of two phases. In the first phase, we take care of the heavy items (if any) by offering them alone until either they are fully sold or no buyer is left. We then take care of the light items by offering them all together in larger bundles. A formal description of this algorithm, which we shall refer to as ``{\bf Alg}", is summarized below.

\begin{algorithm}[h]\caption{A Deterministic Online Hybrid Algorithm}\label{alg-main}
{\bf Phase 1:} Let $[h]=\{1,2,\ldots,h\}$ denote the set of heavy items that are sorted according to their quality (heaviness), i.e., $\theta_1\ge\ldots\ge\theta_h$. Starting from time $t=0$, offer the items in $[h]$ individually and according to their quality order until either all the heavy items have been fully sold, in which case go to Phase 2, or no buyer is left, in which case stop.    

\noindent 
{\bf Phase 2:} At each time $t$, bundle all the available light items together and offer them to buyer $t$ until either we run out of items or no buyer is left. 
\end{algorithm}

Note that since an online algorithm does not know the number of buyers or the sellers' initial inventory levels ahead of time, a competitive online algorithm must carefully balance a trade-off between two scenarios: 1) increasing the chance of selling items by offering larger assortments (hence reducing prices) when there are only a few buyers and many items, and 2) reducing the competition among the sellers by offering smaller assortments (hence increasing prices) when there are many buyers and only a few items. In fact, the desire to balance those two cases is the main reason why the two phases of Algorithm \ref{alg-main} were developed. It is worth noting that Algorithm \ref{alg-main} can be implemented in a fully online fashion and that does not need to know anything about the number of buyers $m$ or the initial inventory levels $\boldsymbol{c}$. At each time $t$, it only needs to know whether there is an arriving buyer and whether an item is still available at that time. Finally, we note that because of the random realization of buyers' choices, the time at which Phase 1 in Algorithm \ref{alg-main} terminates is a random variable $\tau$, which can be at most $\min\{m,\sum_{i=1}^hc_i\}$. The following lemma shows that for all sample paths $w$ that Algorithm \ref{alg-main} has a chance to enter in its second phase (i.e., $\tau(w)\leq m$), the expected revenue obtained during Phase 2 is within a constant factor of the total revenue that light items contribute to the optimal LP value in \eqref{eq:LP-Game}.

\begin{lemma}\label{lemm:phase2-analysis}
Let $\omega$ be any sample path for purchasing items of length $\tau(\omega)\leq m$ that may be realized during Phase 1 of Algorithm \ref{alg-main}. Conditioned on $\omega$, the expected revenue obtained in Phase 2 is at least $\frac{1-\lambda}{2}$ of the total revenue that light items contribute to OPT after time $\tau(\omega)$, where the expectation is with respect to random choices made by buyers during Phase 2. 
\end{lemma}
\begin{proof}
Without loss of generality, we may assume $\tau(\omega)<m$. Otherwise, for $\tau(\omega)=m$, both revenues in the statement of the theorem are equal to zero, and the result trivially holds. Let $\{\boldsymbol{y}^t\}_{t=\tau(\omega)+1}^{m}$ be the optimal offline solution to LP in \eqref{eq:LP-Game} over $[\tau(\omega)+1, m]$. The total contribution of light items $L=[n]\setminus[h]$ to the revenue of OPT after time $\tau(\omega)$ is
\begin{align}\nonumber
D=\sum_{t=\tau(\omega)+1}^{m}\sum_{i\in L}\sum_{S\subseteq [n]}\frac{q_i(S)}{1-q_i(S)} y^t(S).
\end{align}
Now let us consider
\begin{align}\label{eq:LP-residual}
\max &\ \ \sum_{t=\tau(\omega)+1}^m\sum_{i\in L}\sum_{T\subseteq L}q_i(T)x^t(T)\cr 
\mbox{s.t.}  &\ \ \sum_{t=\tau(\omega)+1}^m\sum_{T\subseteq L}q_i(T)x^t(T)\leq c_i \ \ \forall i\in L,\cr 
& \ \ \ \ \ \boldsymbol{x}^t\in \Delta_{2^{|L|}} \ \ t=\tau+1,\ldots,m,
\end{align}
and notice that \eqref{eq:LP-residual} is precisely the LP relaxation upper bound for any online algorithm that can be used during Phase 2 with constant revenues $r_i=1, \forall i\in L$. On the other hand, by Lemma \ref{lemm:basic-Big}, the greedy algorithm used during Phase 2 obtains at least $\frac{1}{2}$ of the optimal value in \eqref{eq:LP-residual}. Thus, if there exists a feasible solution to \eqref{eq:LP-residual} with an objective value of at least $(1-\lambda)D$, then we can conclude that the revenue obtained during Phase 2 is at least $\frac{(1-\lambda)D}{2}$, completing the proof. Thus, in the rest of the proof, we construct a feasible solution to \eqref{eq:LP-residual} with an objective value of at least $(1-\lambda)D$.

Let us fix an arbitrary $t\in [\tau(\omega)+1, m]$ and, for simplicity, we drop the time index $t$. Define $p_i=\sum_{S\subseteq [n]}q_i(S)y(S)$ for $i\in L$, and consider the following LP: 
\begin{align}\label{eq:LP-slack}
\min &\ \ \sum_{i\in L}(\sum_{T\subseteq L}q_i(T)z(T)-p_i)\cr 
\mbox{s.t.}  &\ \ \sum_{T\subseteq L}q_i(T)z(T)-p_i\geq 0 \ \forall i\in L,\cr 
& \ \ \ \boldsymbol{z}\in \Delta_{2^{|L|}}.
\end{align} 
Note that \eqref{eq:LP-slack} is a feasible LP, because using the substitutability property, we have
\begin{align}\nonumber
p_i=\sum_{S\subseteq [n]}q_i(S)y(S)&\leq \sum_{S\subseteq [n]}q_i(S\!\setminus\! [h])y(S)=\sum_{T\subseteq L} \Big(\!\sum_{S\setminus[h]=T}\!\!y(S)\Big)q_i(T),
\end{align} 
and hence $z(T)=\sum_{S\setminus[h]=T}y(S),\ T\subseteq L$ is a feasible solution to \eqref{eq:LP-slack}. 

Next, we show that the optimal value of \eqref{eq:LP-slack} is zero. To derive a contradiction, let $z^*$ be an optimal solution to \eqref{eq:LP-slack} with a strictly positive objective value, and let $i^*$ be an item for which  $\sum_{T\subseteq L}q_{i^*}(T)z^*(T)- p_{i^*}>0$. That means that there exists $T^*\subseteq L$ such that $i^*\in T^*$ and $z^*(T^*)>0$. Note that $T^*$ must contain at least one item other than $i^*$. Otherwise, if $T^*=\{i^*\}$, reducing $z^*(T^*)$ by a small positive amount and increasing $z^*(\emptyset)$ by the same amount will give us another feasible solution to \eqref{eq:LP-slack} with a strictly smaller objective value, contradicting the optimality of $z^*$. Now let us partition the items in $T^*$ into $T^*_1$ and $T^*_2$, where 
\begin{align}\nonumber
&T^*_1=\Big\{i\in T^*: \sum_{T\subseteq L}q_{i}(T)z^*(T)- p_{i}>0\Big\},\cr 
&T^*_2=\Big\{i\in T^*: \sum_{T\subseteq L}q_{i}(T)z^*(T)- p_{i}=0\Big\}.
\end{align}
By the definition of $T^*_1$, that means that there exists a positive number $\epsilon_1>0$ such that reducing $z^*(T^*)$ to $z^*(T^*)-\epsilon_1$, i.e., setting $z^*(T^*)=z^*(T^*)-\epsilon_1$ (by abuse of notation), will preserve the feasibility of all the constraints associated with items $i\in T^*_1$. Note that such a change can only affect the feasibility of the constraints $i\in T^*$ and has no influence on the constraints $i\notin T^*$ (as $q_i(T^*)=0,\forall i\notin T^*$). Unfortunately, the update violates the constraints $i\in T^*_2$ such that in the new solution $\sum_{T\subseteq L}q_{i}(T)z^*(T)- p_{i}=-\epsilon_1q_{i}(T^*), \forall i\in T^*_2$. However, we will show that one can sequentially redistribute the $\epsilon_1$-mass that was removed from $T^*$ to nested subsets of $T^*_2$ and again satisfy all the constraints in $i\in T^*_2$ at equality.

Let $\epsilon_2=\min_{i\in T^*_2} \frac{q_{i}(T^*)\epsilon_1}{q_{i}(T^*_2)}$, and note that by the substitutability property and since $T^*_2\subset T^*$, we have $\epsilon_2<\epsilon_1$. Therefore, if we return an $\epsilon_2$ amount from $\epsilon_1$-mass to $z^*(T^*_2)$, i.e., set $z^*(T^*_2)=z^*(T^*_2)+\epsilon_2$, each constraint in $i\in T_2^*$ becomes ``more" feasible, and at least one constraint (namely the one that achieves $\argmin_{i\in T^*_2} \frac{q_{i}(T^*)}{q_{i}(T^*_2)}\epsilon_1$) is satisfied by equality. Let $T^*_3\subset T^*_2$ be all the items in $T^*_2$ that are not satisfied by equality after the update, i.e.,
\begin{align}\nonumber
T^*_3=\Big\{i\in T^*_2: \sum_{T\subseteq L}q_{i}(T)z^*(T)- p_{i}=\epsilon_2q_{i}(T^*_2)-\epsilon_1q_{i}(T^*)<0\Big\},
\end{align}
and define $\epsilon_3=\min_{i\in T^*_3} \frac{\epsilon_1q_{i}(T^*)-\epsilon_2q_{i}(T^*_2)}{q_{i}(T^*_3)}$. Again by substitutability and since $T^*_3\subset T^*_2\subset T^*$, we have 
\begin{align}\nonumber
\frac{\epsilon_1q_{i}(T^*)-\epsilon_2q_{i}(T^*_2)}{q_{i}(T^*_3)}\leq \frac{(\epsilon_1-\epsilon_2)q_{i}(T^*)}{q_{i}(T^*_3)}\leq \epsilon_1-\epsilon_2,\ \forall i\in T^*_3,
\end{align}
and thus $\epsilon_3<\epsilon_1-\epsilon_2$. Therefore, if we relocate an $\epsilon_3$ amount of the leftover mass $\epsilon_1-\epsilon_2$ to $z^*(T^*_3)$ by setting $z^*(T^*_2)=z^*(T^*_2)+\epsilon_3$, every constraint in $i\in T_3^*$ becomes more feasible, and at least one constraint is tight at equality. By repeating that argument inductively, one can see that the substitutability property means that we always have enough leftover mass to make one more constraint in $T^*_2$ tight so that at the end of this process all the constraints in $T_2^*$ are satisfied at equality. Finally, using $\epsilon=\epsilon_1-\epsilon_2-\epsilon_3-\ldots$ to denote the leftover mass at the end of this process, we can relocate that mass to the empty set by setting $z^*(\emptyset)=z^*(\emptyset)+\epsilon$. The last step does not affect the feasibility of any constraints and only guarantees that the mass conservation is preserved so that $\boldsymbol{z}^*\in \Delta_{2^{|L|}}$. Thus, at the end of the process, we obtain a feasible solution to \eqref{eq:LP-slack} with a strictly smaller objective value than the initial optimal solution $\boldsymbol{z}^*$, a contradiction. Therefore, the optimal value of \eqref{eq:LP-slack} is zero, and there exists $\boldsymbol{z}^*\in \Delta_{2^{|L|}}$ such that 
\begin{align}\nonumber
\sum_{T\subseteq L}q_i(T)z^*(T)=\sum_{S\subseteq [n]}q_i(S)y(S), \forall i\in L. 
\end{align}

As the above argument holds for any $t\in [\tau(\omega)+1,m]$, we obtain a feasible solution $\{\boldsymbol{z}^{*t}\}_{t=\tau(\omega)+1}^m$ to \eqref{eq:LP-residual} that consumes the exact same amount of each resource $i\in L$ that is consumed by the optimal offline solution $\{\boldsymbol{y}^t\}_{t=\tau(\omega)+1}^m$. In particular, the objective value of \eqref{eq:LP-residual} for $\{\boldsymbol{z}^{*t}\}_{t=\tau+1}^m$ equals
\begin{align}\nonumber
\sum_{t=\tau(\omega)+1}^m\sum_{i\in L}\sum_{T\subseteq L}q_i(T)z^{*t}(T)=\sum_{t=\tau(\omega)+1}^m\sum_{i\in L}\sum_{S\subseteq [n]}q_i(S)y^t(S)\ge  (1-\lambda)\sum_{t=\tau(\omega)+1}^{T}\sum_{i\in L}\sum_{S\subseteq [n]}\frac{q_i(S)}{1-q_i(S)} y^t(S)=(1-\lambda)D,
\end{align}
where the inequality holds because for every light item $i\in L$, $\frac{q_i(S)}{1-q_i(S)}\leq \frac{q_i(S)}{1-\lambda}$. 
\end{proof}


\begin{definition}\label{def-omega}
We let $\Omega$ denote the set of all sample paths of \emph{arbitrary} length that can be realized by offering heavy items individually and according to their quality order until all the heavy items $[h]$ have been fully sold.\footnote{This implies that any $\omega\in \Omega$ sells all the heavy items, albeit some to ``virtual" buyers that come after time $m$. In other words,  $\Omega$ is the set of sample paths of arbitrary length that can be realized by executing Phase 1, assuming that there are infinitely many buyers.} Moreover,  we let $\tau_i:\Omega\to Z_+$ be a random variable, where $\tau_i(\omega)$ denotes the first time that item $i\in[h]$ is fully sold over the sample path $\omega$. We note that $\tau_1\leq \ldots\leq \tau_h$. 
\end{definition}

\begin{definition}
We define $A\subseteq \Omega$ to include sample paths for which all the heavy items are fully sold during Phase 1, i.e.,  $A=\{\omega\in \Omega: \tau_h(\omega)\leq m\}$. Moreover,  we let $\bar{A}=\Omega\!\setminus\! A=\{\omega\in \Omega: \tau_h(\omega)> m\}$.  Intuitively,  the sets $A$ and $\bar{A}$ contain all sample paths for which Algorithm \ref{alg-main} gets a chance and does not get a chance to enter its second phase, respectively.
\end{definition}

The reason for introducing the extended sample space $\Omega$ in Definition \ref{def-omega} is to assure that all the subsequent random variables are defined over the same probability space $(\Omega, \mathbb{P})$.  Such a definition naturally allows us to compute the realization probability of an actual sample path $\omega_0$ during the execution of Phase 1 of Algorithm \ref{alg-main}.  More precisely,  given an actual sample path $\omega_0$ of length at most $m$ that is realized during the execution of Phase 1 of Algorithm \ref{alg-main}, either all the heavy items are sold over $\omega_0$, in which case $\omega_0\in A$, and so the actual and extended probabilities are the same. Otherwise,  the actual probability that $\omega_0$ is realized equals the sum of probabilities of all extended sample paths in $\bar{A}$ whose first $m$ time instances coincide with $\omega_0$.



\begin{lemma}\label{eq:OPT-ALG-A}
Let $\mathbb{E}[R({\rm Alg})|A]$ denote the expected revenue of Algorithm \ref{alg-main} over all the sample paths in $A$. Then, we have ${\rm OPT}\leq \frac{2}{1-\lambda}\mathbb{E}[R({\rm Alg})|A]$.
\end{lemma}
\begin{proof}
Given an arbitrary sample path $\omega\in A$, the revenue obtained during Phase 1 is equal to $R(\mbox{Phase 1}|\omega)=\sum_{i=1}^{h}\frac{c_i}{1-q_i(\{i\})}$. The reason is that over that sample path, all the $c_i$ units of item $i\in[h]$ are fully sold at an equilibrium price of $\frac{1}{1-q_i(\{i\})}$ (as heavy items are offered individually). Thus,
\begin{align}\label{eq:alg-bound-A}
\mathbb{E}[R(\mbox{Alg})|A]=\sum_{i=1}^{h}\frac{c_i}{1-q_i(\{i\})}+\mathbb{E}[R(\mbox{Phase 2})|A].
\end{align}     
Now, given any sample path $\omega\in A$ and using Lemma \ref{lemm-heavy-two}, we can upper-bound the revenue of OPT up to time $\tau_1(\omega)$ as 
\begin{align}\nonumber
\sum_{t=1}^{\tau_1(\omega)}\sum_S R(S)y^t(S)&\leq \sum_{t=1}^{\tau_1(\omega)}\sum_S\frac{f(\lambda)q_1(\{1\})}{1-q_1(\{1\})}y^t(S)=\frac{f(\lambda)q_1(\{1\})}{1-q_1(\{1\})}\tau_1(\omega),
\end{align}
where we recall that $\tau_i(\omega)$ denotes the first time that item $i\in[h]$ is fully sold along the sample path $\omega$.  Also, the total contribution of item $1$ to the value of OPT is
\begin{align}\nonumber
\sum_{t=1}^{m}\sum_S\frac{q_1(S)}{1-q_1(S)}y^t(S)&\leq \sum_{t=1}^{m}\sum_S\frac{q_1(S)}{1-q_1(\{1\})}y^t(S)\leq \frac{c_1}{1-q_1(\{1\})},
\end{align}  
where the first inequality is by the substitutability property, and the second inequality is true because $\{y^t(S)\}$ is a feasible solution to \eqref{eq:LP-Game}. As we have considered the total contribution of item $1$ to OPT, we can  remove item $1$ from all the assortments offered by OPT and upper-bound the remaining revenue of OPT over $[\tau_1(\omega)+1, \tau_2(\omega)]$ as  
\begin{align}\nonumber
&\sum_{t=\tau_1(\omega)+1}^{\tau_2(\omega)}\sum_S \Big(R(S)-\frac{q_1(S)}{1-q_1(S)}\Big)y^t(S)\leq \sum_{t=\tau_1(\omega)+1}^{\tau_2(\omega)}\sum_S R(S\!\setminus\!\{1\})y^t(S)\cr 
&\qquad\qquad\leq \sum_{t=\tau_1(\omega)+1}^{\tau_2(\omega)}\sum_S\frac{f(\lambda)q_2(\{2\})}{1-q_2(\{2\})}y^t(S)=\frac{f(\lambda)q_2(\{2\})}{1-q_2(\{2\})}(\tau_2(\omega)-\tau_1(\omega)),
\end{align}
where the first inequality is by the substitutability property, and the second inequality is from Lemma \ref{lemm-heavy-two} (as the heaviest item in $S\setminus\{1\}$ at best can be item $2$). Similarly, the total contribution of item $2$ to the value of OPT is
\begin{align}\nonumber
\sum_{t=1}^{m}\sum_S\frac{q_2(S)}{1-q_2(S)}y^t(S)&\leq \sum_{t=1}^{m}\sum_S\frac{q_2(S)}{1-q_2(\{2\})}y^t(S)\leq \frac{c_2}{1-q_2(\{2\})}.
\end{align} 
Again, we can safely remove item $2$ from all the assortments in OPT and repeat the same process to show that for any $i\in [h]$ and any sample path $\omega$ with $\tau_i(\omega)\leq m$, the revenue of OPT during time instances $[1,  \tau_i(\omega)]$ is bounded above by
\begin{align}\label{eq:tau-OPT-bound}
\sum_{j=1}^{i}\frac{f(\lambda)q_j(\{j\})}{1-q_j(\{j\})}(\tau_j(\omega)-\tau_{j-1}(\omega))+\sum_{j=1}^{i-1} \frac{c_j}{1-q_j(\{j\})}. 
\end{align}
At the end of the above process, all the heavy items have been completely removed from the assortments offered by OPT, and we will be left only with light items whose contribution to the OPT is $\sum_{t=\tau_h(\omega)+1}^{m}\sum_{S}\sum_{i\in L}\frac{q_i(S)}{1-q_i(S)}y^t(S)$. But from Lemma \ref{lemm:phase2-analysis}, that value is at most $\frac{2}{1-\lambda}\mathbb{E}[R(\mbox{Phase 2})|\omega]$, where the expectation is taken with respect to buyers' random choices during the second phase of the algorithm. Thus, for every sample path $\omega\in A$, we have shown that OPT is at most
\begin{align}\nonumber
\mbox{OPT}&\leq \sum_{i=1}^{h}\frac{f(\lambda)q_i(\{i\})}{1-q_i(\{i\})}(\tau_i(\omega)-\tau_{i-1}(\omega))+\sum_{i=1}^{h} \frac{c_i}{1-q_i(\{i\})}+\frac{2}{1-\lambda}\mathbb{E}[R(\mbox{Phase 2})|\omega],
\end{align}  
where by convention $\tau_0=0$. Taking the conditional expectation $\mathbb{E}[\cdot|A]$ from the above inequality over all $\omega\in A$, we get
\begin{align}\nonumber
\mbox{OPT}&\leq \sum_{i=1}^{h}\frac{f(\lambda)q_i(\{i\})}{1-q_i(\{i\})}\mathbb{E}[\tau_i-\tau_{i-1}|A]+\sum_{i=1}^{h} \frac{c_i}{1-q_i(\{i\})}+\frac{2}{1-\lambda}\mathbb{E}[R(\mbox{Phase 2})|A].
\end{align}
Finally, we note that $\mathbb{E}[\tau_i-\tau_{i-1}|A]\leq\mathbb{E}[\tau_i-\tau_{i-1}]= \frac{c_i}{q_i(\{i\})}$.\footnote{Note that conditioning on the event $A$ puts an upper bound on the random variable $\tau_i-\tau_{i-1}$, hence only resulting in a lower expectation.} If we combine this relation with the above inequality, we get
 \begin{align}\nonumber
\mbox{OPT}&\leq \sum_{i=1}^{h} \frac{(1+f(\lambda))c_i}{1-q_i(\{i\})}+\frac{2}{1-\lambda}\mathbb{E}[R(\mbox{Phase 2})|A]\leq  \frac{2}{1-\lambda}\mathbb{E}[R(\mbox{Alg})|A],
\end{align}
where the second inequality is due to \eqref{eq:alg-bound-A} and the fact that $1+f(\lambda)\leq \frac{2}{1-\lambda}, \forall \lambda\ge \frac{1}{2}$. 
\end{proof}

In the following, we prove the main result of this section, which is a constant competitive ratio for the online assortment problem. The main idea of the proof is to show that if Algorithm \ref{alg-main} does not get a chance to enter its second phase, the reason is that the number of buyers $m$ is small.  Otherwise, the revenue obtained from Phase 1 is sufficiently large compared to the optimal value of the LP \eqref{eq:LP-Game}. 

\begin{theorem}\label{them:online}
Algorithm \ref{alg-main} is a constant competitive algorithm for the online assortment problem with homogeneous buyers. In particular, $\frac{\mathbb{E}[R({\rm Alg})]}{{\rm OPT}}\ge 0.057$.   
\end{theorem}
\begin{proof}
Let $c=\sum_{i=1}^{h}c_i$ be the total number of heavy items, and define $B\subseteq \bar{A}$ to be the subset of sample paths in $\bar{A}$ that sell at least $a$ heavy items during Phase 1 (i.e.,  during the first $m$ time instances), for some $a\leq \min\{m,c\}$ to be determined later. More precisely, let $X:\Omega\to \mathbb{Z}_+$ be a random variable, where $X(\omega)$ denotes the number of heavy items sold during the first $m$ time instances of the sample path $\omega$, and note that $X\leq m$. Then, $B=\{\omega\in \Omega: X(\omega)>a, \tau_h(\omega)>m\}$, and we can write
\begin{align}\label{eq:B-bound-B}
\mathbb{P}(B)&=\mathbb{P}(\tau_h>m)\mathbb{P}(X>a|\tau_h>m)\cr 
&=\mathbb{P}(\bar{A})\Big(1-\mathbb{P}\big(m-X\ge m-a|\bar{A}\big)\Big)\cr
&\geq \mathbb{P}(\bar{A})\Big(1-\frac{\mathbb{E}[m-X|\bar{A}]}{m-a}\Big)\cr 
&\ge \mathbb{P}(\bar{A})\Big(1-\frac{\mathbb{E}[m-X]}{\mathbb{P}(\bar{A})(m-a)}\Big),
\end{align}   
where the first inequality is by Markov's inequality, and the second inequality uses $\mathbb{E}[m-X]\ge \mathbb{E}[m-X|\bar{A}]\mathbb{P}(\bar{A})$. On the other hand, we know that at each time a heavy item $i$ is offered, it is sold independently with a probability of at least $q_i(\{i\})\ge \lambda$. Therefore, $\mathbb{E}[X]\ge \min\{\lambda m,c\}$. By combining this relation with \eqref{eq:B-bound-B}, we obtain
\begin{align}\label{eq:A-B}
\mathbb{P}(B)\ge \mathbb{P}(\bar{A})+\frac{\min\{\lambda m,c\}-m}{m-a}.
\end{align}
Now consider an arbitrary sample path $\omega\in B$ that sells $X(\omega)\ge a$ heavy items during its first $m$ time instances. To upper-bound the revenue of OPT, let $i\in [h]$ be the last heavy item that is sold during the first $m$ time instances of $\omega$.  Using an argument similar to that in the proof of Lemma \ref{eq:OPT-ALG-A},  the revenue of OPT during time instances $[1, X(\omega)]$ can be upper-bounded using \eqref{eq:tau-OPT-bound} with the specific sample path $\hat{\omega}$ for which $\tau_0(\hat{\omega})=0,  \tau_j(\hat{\omega})=\tau_{j-1}(\hat{\omega})+c_j\ \forall j\in [i-1],  \tau_i(\hat{\omega})=X(\omega)$, to obtain
\begin{align}\nonumber
\sum_{j=1}^{i-1}\frac{f(\lambda)q_j(\{j\})}{1-q_j(\{j\})}c_j+\frac{f(\lambda)q_i(\{i\})}{1-q_i(\{i\})}(X(\omega)-\sum_{j=1}^{i-1}c_j)+\sum_{j=1}^{i-1} \frac{c_j}{1-q_j(\{j\})}< (1+f(\lambda))R(\mbox{Alg}|\omega),
\end{align}
where the inequality holds because $R(\mbox{Alg}|\omega)= \sum_{j=1}^{i-1}\frac{c_j}{1-q_j(\{j\})}\!+\!\frac{X(\omega)-\sum_{j=1}^{i-1}c_j}{1-q_i(\{i\})}$,\footnote{Note that $R(\mbox{Alg}|\omega)$ is only determined by the first $m$ time instances of $\omega$. Moreover, since $\omega\in B$, the algorithm does not even get a chance to enter its second phase to sell light items.} and $q_j(\{j\})<1$. We can then remove the heavy items $1,2,\ldots,i-1$ from all the assortments offered by the OPT after time $X(\omega)$, in which case the remaining assortments offered by OPT over $[X(\omega)+1, m]$ can only contain item $i$ or lighter items. Based on Lemma \ref{lemm-heavy-two}, that means that the revenue of each remaining assortment in OPT is at most $f(\lambda)$ times  the revenue of the individual assortment $\{i\}$. Thus, the remaining revenue obtained by OPT over $[X(\omega)+1, m]$ is at most $(m-X(\omega))\frac{f(\lambda)q_i(\{i\})}{1-q_i(\{i\})}$. Moreover, as Algorithm \ref{alg-main} sells $X(\omega)$ items of type $i$ or heavier during the first $m$ time instances of $\omega$, we have $\frac{1}{1-q_i(\{i\})}\leq \frac{R({\rm Alg}|\omega)}{X(\omega)}$. Putting it all together, for any $\omega\in B$, we have shown that
\begin{align}\nonumber
\mbox{OPT}&\leq (1+f(\lambda)) R(\mbox{Alg}|\omega)+(m-X(\omega))\frac{f(\lambda)R(\mbox{Alg}|\omega)}{X(\omega)}\cr 
&=(\frac{f(\lambda)m+X(\omega)}{X(\omega)}) R(\mbox{Alg}|\omega)\leq(\frac{f(\lambda)m+a}{a}) R(\mbox{Alg}|\omega),
\end{align}  
where the last inequality holds because $X(\omega)\ge a, \forall \omega\in B$.  If we take the conditional expectation $\mathbb{E}[\cdot|B]$ from both sides of the above inequality, we get $\mbox{OPT}\leq (\frac{f(\lambda)m+a}{a}) \mathbb{E}[R(\mbox{Alg})|B]$. Using this relation together with Lemma \ref{eq:OPT-ALG-A} and \eqref{eq:A-B}, we get
\begin{align}\label{eq:OPT-Alg-Big}
\mathbb{E}[R(\mbox{Alg})]&\ge \mathbb{E}[R(\mbox{Alg})|A]\cdot \mathbb{P}(A)+\mathbb{E}[R(\mbox{Alg})|B]\cdot \mathbb{P}(B)\cr 
&\ge  (\frac{1-\lambda}{2})\mbox{OPT}\cdot \mathbb{P}(A)+\frac{\mbox{OPT}}{(\frac{f(\lambda)m+a}{a})}\cdot\mathbb{P}(B)\cr 
&\geq \Big( (\frac{1-\lambda}{2})\mathbb{P}(A)+\frac{\mathbb{P}(\bar{A})+\frac{\min\{ \lambda m,c\}-m}{m-a}}{(\frac{f(\lambda)m+a}{a})}\Big)\mbox{OPT}\cr 
&=\Big(\mathbb{P}(A)(\frac{1-\lambda}{2}-\frac{a}{f(\lambda)m+a})+\frac{(\min\{\lambda m,c\}-\!a)a}{(f(\lambda)m+a)(m-a)}\Big)\mbox{OPT}\cr 
&=\Big(\mathbb{P}(A)(\frac{1-\lambda}{2}-\frac{x}{f(\lambda)+x})+\frac{(\min\{\lambda,\frac{c}{m}\}-x)x}{(f(\lambda)+x)(1-x)}\Big)\mbox{OPT},
\end{align}
where $x=\frac{a}{m}\ge 0$. Moreover, by Markov's inequality, 
\begin{align}\label{eq:Markov-chernof}
\mathbb{P}(\bar{A})&=\mathbb{P}(\tau_h>m)\leq \frac{\mathbb{E}[\tau_h]}{m}=\frac{\sum_{i=0}^h\mathbb{E}[\tau_i-\tau_{i-1}]}{m}=\sum_{i=0}^h\frac{c_i}{mq_i(\{i\})}\leq \frac{c}{\lambda m},
\end{align}
and thus $\mathbb{P}(A)\ge \max\{0, 1-\frac{c}{\lambda m}\}$. Let $y=\frac{c}{m}\ge 0$, and note that $x\leq \min\{y,1\}$, depending on whether $c\leq m$ or $c>m$.  If we substitute the bound on $\mathbb{P}(A)$ into \eqref{eq:OPT-Alg-Big} and maximize the result over $x\leq \min\{y,1\}$ while minimizing it over $y>0$, the competitive ratio of the algorithm for a fixed threshold $\lambda$ is at least
\begin{align}\nonumber
g(\lambda)=\min_{y\ge 0}\ \max_{0\leq x\leq \min\{y,1\}}\Big\{\max\{0, 1-\frac{y}{\lambda}\}\big(\frac{1-\lambda}{2}-\frac{x}{f(\lambda)+x}\big)+\frac{(\min\{\lambda,y\}-x)x}{(f(\lambda)+x)(1-x)}\Big\}.
\end{align}  
Finally, by maximizing $g(\lambda)$ over $\frac{1}{2}\leq \lambda\leq 1$, we get $\max_{\lambda\in [\frac{1}{2}, 1]}g(\lambda)\ge 0.057$, that is obtained for $\lambda=0.63$ (see Appendix I). This shows that the competitive ratio of the Algorithm \ref{alg-main} with threshold  $\lambda=0.63$ is at least $\frac{\mathbb{E}[R({\rm Alg})]}{{\rm OPT}}\ge 0.057$. 
\end{proof}

In the proof of Theorem \ref{them:online}, one could further leverage the i.i.d. property of buyers and use tighter Chernoff bounds rather than Markov's inequality in \eqref{eq:Markov-chernof} to improve the competitive ratio. However, for the sake of simplicity, we did not follow that path. Instead, in Appendix II, we have conducted some numerical experiments to illustrate the outperformance of the hybrid algorithm beyond the theoretical guarantee of Theorem \ref{them:online}. One advantage of using Markov's inequality rather than Chernoff bound in our analysis is that it can be used to analyze the competitive ratio even if there are statistical correlations among buyers. In fact, the above analysis of the hybrid Algorithm \ref{alg-main} holds even under a more general setting as long as equilibrium demands and revenues satisfy the substitutability property given in Lemma \ref{lemm:subtitute} and the revenue approximation given in Lemma \ref{lemm-heavy-two}.


Finally, we mention a potential extension of our results to \emph{perishable} items when sellers have time limits to sell their items.  For instance, one can consider a scenario in which the sellers lose an item even if the buyer did not buy their items. One way of modeling such a setting is to modify the sellers' revenues to account for the negative effect of losing items even if they are not purchased.  Given an item $i$, let $\beta_i$ be the fixed cost of producing it. Then, seller $i$'s revenue can be computed as $\bar{R}_i(S)=\bar{p}_i(S)\bar{q}_i(S)-\beta_i(1-\bar{q}_i(S))$, where $S$ is the assortment offered by the platform, $\bar{p}_i(S)$ is the offered price for item $i$, and $\bar{q}_i(S)=\frac{\exp(\theta_i-\bar{p}_i(S))}{\sum_{j\in S\cup\{0\}}\exp(\theta_j-\bar{p}_j(S))}$ is the purchase probability of item $i$. Here, the first term $\bar{p}_i(S)\bar{q}_i(S)$ captures the revenue of selling item $i$ and the negative term $-\beta_i(1-\bar{q}_i(S))$ captures the expected loss of the seller if item $i$ was not purchased. The reason is that seller $i$ incurs a cost of $\beta_i$ for producing a unit of item $i$,  and if the item was not sold to the buyer (which happens with probability $1-\bar{q}_i(S)$), then the expected loss for the seller would be $\beta_i(1-\bar{q}_i(S))$. Now, one can again compute the equilibrium prices for this modified revenue function and show that equilibrium prices always exist and are given by $p_i(S)=\frac{1}{1-q_i(S)}-\beta_i$, where $q_i(S)$ is the equilibrium demand for item $i$ given the assortment $S$.\footnote{Intuitively, If a seller knows that he will lose his item regardless of whether it is purchased or not, he is willing to offer that item at a lower price to increase the chance of selling it.} In particular, the stage equilibrium revenue for seller $i\in [n]$ becomes $R_i(S)=\frac{q_i(S)}{1-q_i(S)}-\beta_i$. As a result, the platform's objective function in this new setting differs from the original objective function \eqref{eq:rev-objective} by only an additive constant $-m\sum_i\beta_i$. Therefore, from an optimization perspective, the online platform is essentially solving the same optimization problem as before, and the competitive ratio analysis of Algorithm \ref{alg-main} remains valid. 


\section{Offline Generalized Bertrand Game with Heterogeneous Buyers}\label{sec-offline}

As we showed earlier, the online assortment problem with heterogeneous buyers does not admit a constant competitive algorithm. To compensate for the nonconstant competitive ratio of heterogeneous buyers, one might consider identical buyers that choose according to more general choice models, or assume heterogeneous buyers whose evaluations belong to specific distributions. However, in this section, we take an alternative approach, and consider an offline market with arbitrary heterogeneous buyers. Analyzing such a generalized offline market is important for several reasons. i) It provides a solution to the case in which multiple heterogeneous buyers simultaneously arrive in the market. ii) While offline markets with multiple buyers/sellers under oligopolistic competitions, such as network Cournot competition, have been well studied \cite{bimpikis2019cournot,lin2017networked,abolhassani2014network}, the existing results for Bertrand competitions over general bipartite networks are very limited. iii) Finally, analyzing the optimal segmentation of the generalized Bertrand game allows the platform to improve its revenue by effectively clustering heterogeneous buyers into consistent classes and matching them to their favorite sellers, hence improving buyers' satisfaction.

In the offline generalized Bertrand game, we again assume that there is a set of $[n]$ sellers in which seller $i$ has $c_i\in \mathbb{Z}_+$ units of product $i$. Moreover, there is a set of $[m]$ buyers, in which buyer $k$'s evaluations of product qualities are given by real numbers $\theta_{1k},\ldots,\theta_{nk}$. As before, we also consider a ``no-purchase" item $0$ with a normalized price $p_0=0$, and it is assumed that all buyers' evaluations for the no-purchase item are $\theta_{0k}=0\ \forall k\in [m]$. Now, if we use $\boldsymbol{p}=(p_1,\ldots,p_n)\in\mathbb{R}^n_+$ to denote the sellers' posted prices, the expected utility derived by seller $i$ is given by
\begin{align}\label{eq:offline-utility} 
\mathcal{U}_i(p_i,p_{-i})=p_i\min\{\sum_{k=1}^{m}q_{ik}, c_i\},
\end{align} 
where $q_{ik}=\frac{\exp(\theta_{ik}-p_i)}{1+\sum_{j=1}^n\exp(\theta_{jk}-p_{j})}$ is the MNL probability that buyer $k$ will purchase item $i$ (i.e., expected demand). Note that since seller $i$ has at most $c_i$ units of item $i$, the expected revenue \eqref{eq:offline-utility} that seller $i$ can derive at a price $p_i$ is at most $p_ic_i$, even though he or she may receive more demand than the inventory $c_i$. Therefore, the utility functions in \eqref{eq:offline-utility} define a noncooperative game among the sellers in which each seller wants to set a price for his or her item to maximize the expected utility. 

Unlike the single-buyer Bertrand game ($m=1$), which is a potential game (see Remark \ref{rem:potential}) and hence admits a pure Nash equilibrium \cite{monderer1996potential}, the above generalized Bertrand game does not seem to admit a potential function. Therefore, an immediate question concerning this generalized Bertrand game is whether it admits a \emph{pure-strategy} Nash equilibrium. In the following, we will show that under a mild assumption on the buyers' evaluations, the generalized Bertrand game admits a pure-strategy Nash equilibrium over arbitrary bipartite graphs.
  
\begin{assumption}\label{ass:consistent}
The generalized Bertrand game is called \emph{consistent} if sellers' posted prices can incentivize a buyer by up to $91\%$ to buy any certain product.
\end{assumption}

Assumption \ref{ass:consistent} implies that buyers' evaluations of product qualities are consistent and lie within a certain range of each other. This is a reasonable assumption for several reasons, particularly in platforms with side information (e.g., Amazon.com) such that buyers have access to product reviews. This reason is that if a product truly has a certain quality, it is unlikely that one buyer will evaluate its quality extremely high while the others evaluate its quality extremely low. Moreover, buyers are often not fully determined to buy a product and merely explore the market for suitable alternatives. As a result, sellers' posted prices cannot incentivize a buyer $100\%$ to buy a product. For instance, one way to assure that Assumption \ref{ass:consistent} holds is to assume that buyers' evaluations are bounded above by $\theta_{ik}\leq 2.3, \forall i, k$. In that case, the demand of buyer $k$ for item $i$ is at most $q_{ik}=\frac{\exp(\theta_{ik}-p_i)}{1+\sum_{j=1}^n\exp(\theta_{jk}-p_{j})}\leq \frac{\exp(\theta_{ik})}{1+\exp(\theta_{ik})}\leq 0.91$. Thus, regardless of posted prices, a buyer will not purchase an item with a probability of more than $91\%$ from any seller.

\begin{theorem}\label{thm:pure}
Under market consistency Assumption \ref{ass:consistent}, the generalized Bertrand game over general bipartite graphs admits a pure-strategy Nash equilibrium. 
\end{theorem}
\begin{proof}
First, let us assume that the Bertrand game is captured by a complete bipartite graph, meaning that every seller $i\in [n]$ is visible to every buyer $k\in [m]$ with a quality evaluation $\theta_{ik}$. We will show that the utility function of each player $i\in [n]$ is quasiconcave with respect to its own decision variable $p_i$. This quasiconcavity, in view of \cite{baye1993characterizations} and the fact that the utility functions $\mathcal{U}_i$ are continuous and players' strategy sets are convex and compact, implies that the generalized Bertrand game admits a pure-strategy Nash equilibrium.\footnote{Although the strategy set of a player is $[0, \infty)$, it can be shown that without loss of generality, all the players must choose their prices in the compact set $[0, \theta]$, where $\theta=\max_{i,k}|\theta_{ik}|$.} 

To establish the quasiconcavity of the utilities, it is enough to show that $u_i(p_i,p_{-i})=p_i\sum_{k=1}^{m}q_{ik}$ is a quasiconcave function of $p_i$. As $p_ic_i$ is a linear (and hence quasiconcave) function and the pointwise minimum of two functions preserves quasiconcavity \cite{boyd2004convex}, we conclude that $\mathcal{U}_i(p_i,p_{-i})$ is also quasiconcave in $p_i$. Therefore, in order to show the quasiconcavity of $u_i(p_i,p_{-i})$ we use the equivalent condition for differentiable functions \cite{boyd2004convex} to show that for every fixed $p_{-i}$, if $\frac{\partial u_i}{\partial p_i}(p_i,p_{-i})=0$, then $\frac{\partial^2 u_i}{\partial p^2_i}(p_i,p_{-i})< 0$. By the definition of choice probabilities, a simple calculation shows that
\begin{align}\nonumber
\frac{\partial q_{jk}}{\partial p_i}=\begin{cases}
q_{ik}^2-q_{ik} \ &\mbox{if }\ \ j=i,\\
q_{ik}q_{jk} \ &\mbox{if }\ \ j\neq i.
\end{cases}
\end{align}
Therefore, 
\begin{align}\label{eq:derivatives}
&\frac{\partial u_i}{\partial p_i}=\sum_{k=1}^m\big(q_{ik}-p_i(q_{ik}-q^2_{ik})\big),\cr 
&\frac{\partial^2 u_i}{\partial p^2_i}=\sum_{k=1}^m(q_{ik}^2-q_{ik})(2+2p_iq_{ik}-p_i).
\end{align}
Let $p_i=\frac{\sum_{k=1}^mq_{ik}}{\sum_{k=1}^{m}(q_{ik}-q^2_{ik})}$ be the solution to $\frac{\partial u_i}{\partial p_i}=0$. If we put that relation into \eqref{eq:derivatives}, and define $P=\sum_{\ell=1}^mq_{i\ell}$, $Q=\sum_{\ell=1}^mq^2_{i\ell}$, and $R=\sum_{\ell=1}^mq^3_{i\ell}$, we can write
\begin{align}\label{eq:P-Q}
\frac{\partial^2 u_i}{\partial p^2_i}&=\sum_{k=1}^m(q_{ik}^2-q_{ik})\Big(2+2\frac{\sum_{\ell=1}^mq_{i\ell}}{P-Q}q_{ik}-\frac{\sum_{\ell=1}^mq_{i\ell}}{P-Q}\Big)\cr 
&=\sum_{k=1}^m\frac{q_{ik}-q_{ik}^2}{Q-P}\Big(2\sum_{\ell=1}^{m}(q_{i\ell}-q^2_{i\ell})+2\sum_{\ell=1}^mq_{i\ell}q_{ik}-\sum_{\ell=1}^mq_{i\ell}\Big)\cr 
&=\frac{-1}{P-Q}\sum_{k=1}^m\sum_{\ell=1}^{m}(q_{ik}-q_{ik}^2)(q_{i\ell}+2q_{i\ell}q_{ik}-2q^2_{i\ell})\cr 
&=\frac{-1}{P-Q}\sum_{k=1}^m\sum_{\ell=1}^{m}q_{ik}q_{i\ell}(1+q_{ik}-2q_{i\ell}+2q_{ik}q_{i\ell}-2q_{ik}^2)\cr 
&=\frac{-1}{P-Q}\Big(P^2+PQ-2PQ+2Q^2-2PR\Big)\cr 
&=\frac{-1}{P-Q}\Big(P^2-PQ+2Q^2-2PR\Big).
\end{align}  
As $q_{i\ell}\in (0,1), \forall i, \ell$, we have $R< Q<P$. Moreover, by Assumption \ref{ass:consistent}, no player $i$ dominates the market by taking more than $91\%$ of the demand of any buyer, i.e., $q_{i\ell}<0.91, \forall i, \ell$. That implies that $R\leq 0.91Q<\frac{2\sqrt{2}-1}{2} Q$. Thus,
\begin{align}\label{eq:P-Q-app}
P^2\!-\!PQ\!+\!2Q^2\!-\!2PR&> P^2\!-\!PQ\!+\!2Q^2\!-\!(2\sqrt{2}\!-\!1)PQ=(P-\sqrt{2}Q)^2\ge 0.
\end{align}
If we use \eqref{eq:P-Q-app} in \eqref{eq:P-Q}, it is easy to see that $\frac{\partial^2 u_i}{\partial p^2_i}<0$, which shows that $u_i(p_i,p_{-i})$ is a quasiconcave function of $p_i$.

Finally, if the market is not captured by a complete bipartite graph such that seller $i$ is visible only to a subset $N_i\subset [m]$ of buyers, then one can carry over all the above analysis by replacing the above summations over $k\in [m]$ with the summations over $k\in N_i$ (or, by assuming $\theta_{ik}\to -\infty$ for every pair of buyer-seller $(k,i)$ that are not visible to each other).
\end{proof}

Unfortunately, because of a highly nonlinear structure of the choice probabilities together with the sellers' capacity constraints, it likely no possible to obtain a closed-form solution for the equilibrium prices of the generalized Bertrand game. However, it is known that quasiconcave games with a unique Nash equilibrium point offer many nice properties, and that lots of simple iterative learning and adjustment rules converge to that equilibrium \cite{rosen1965existence,even2009convergence}. Finally, we would like to mention that the existence of capacity constraints can only push the equilibrium prices higher than they are in the uncapacitated case. The reason is that given an equilibrium price vector $\boldsymbol{p}$, if item $i$ does not have a capacity constraint, then the equilibrium price for that item is $p_i=\frac{\sum_{k=1}^mq_{ik}}{\sum_{k=1}^{m}(q_{ik}-q^2_{ik})}$. On the other hand, if we impose a capacity of $c_i$ on item $i$, either $\sum_{k=1}^{m}q_{ik}\leq  c_i$, in which case the equilibrium price $p_i$ remains as in the uncapacitated case, or $\sum_{k=1}^{m}q_{ik}> c_i$. In the latter case, seller $i$ can strictly increase its utility from $p_ic_i$ to some $(p_i+\epsilon) c_i$ by increasing its price to match its supply $c_i$ with its demand $\sum_{k=1}^{m}q_{ik}$.\footnote{Note that such a matching is possible because $q_{ik}, k\in [m]$ are monotonically decreasing and are continuous functions of the price $p_i$.} Therefore, we have the following corollary:
\begin{corollary}\label{corr:capacity}
At any pure Nash equilibrium of the generalized Bertrand game, the total demand received by a seller $i$ is at most $c_i$, i.e., $\sum_{k=1}^{m}q_{ik}\leq c_i$.
\end{corollary}

Unfortunately, characterizing the amount of increase of equilibrium prices caused by capacity constraints is a complicated function of all other parameters, and that makes it difficult for to obtain a closed-form solution. Therefore, in the next section, we take a different approach to improving the revenue obtained at a Nash equilibrium of the generalized Bertrand game, by segmenting the market into smaller submarkets with easily computable equilibrium prices.

\subsection{Optimal Segmentation of the Generalized Bertrand Game}
In this section, we consider optimal market segmentation for the generalized Bertrand game under consistency Assumption \ref{ass:consistent} (which guarantees the existence of a pure Nash equilibrium). In the market-segmenting problem, the goal is to partition the set of buyers/sellers into smaller pools $\{\mathcal{P}_r\}_{r\in I}$ in which the sellers in each pool $\mathcal{P}_r$ are visible only to the buyers of the same pool. In particular, we are interested in a partitioning that achieves the maximum revenue $\sum_{r\in I}\mbox{Rev}(\mathcal{P}_r)$, where $\mbox{Rev}(\mathcal{P}_r)$ denotes the revenue obtained from the sellers in $\mathcal{P}_r$, given that each pool $\mathcal{P}_r$ is operating at its equilibrium prices and demands. In fact, it has been shown that in certain markets with simpler exogeneous supply-demand curves, segmenting the market into smaller pools can significantly improve the revenue/welfare derived from those markets \cite{banerjee2017segmenting}. In fact, market segmentation can be viewed as an extension of the single-assortment problem to a multi-assortment problem. More precisely, instead of having the platform to offer only one assortment to all the buyers, the platform first partitions the buyers into different clusters based on their preferences and then offers a distinct assortment to each cluster.  Following that idea, we consider the optimal segmenting problem for the generalized Bertrand game and provide a simple approximation algorithm for its optimal segmentation.

\begin{theorem}\label{thm:pool}
Under market consistency Assumption \ref{ass:consistent} and bounded evaluations $\theta_{ik}\in [0, \theta] \ \forall i,k$, where $\theta$ is a constant, there is an $O(\log m)$-approximation for the optimal segmentation of the generalized Bertrand game.
\end{theorem}
\begin{proof}
Let $\mathcal{P}$ be an arbitrary pool in the optimal segmentation; $\mathcal{P}$ contains a subset of sellers $N_\mathcal{P}\subseteq [n]$ and a subset of buyers $M_\mathcal{P}\subseteq [m]$. Under Assumption \ref{ass:consistent} and using Theorem \ref{thm:pure}, we know that this pool equilibriates at some prices $p_i, i\in N_\mathcal{P}$ and demands $q_{ik}\in [0, 0.91], i\in N_\mathcal{P}, k\in M_\mathcal{P}$. Now we consider two cases. If $\sum_{k\in M_\mathcal{P}} q_{ik}<c_i$, then at the equilibrium, we have $\frac{\partial}{\partial p_i}(p_i\sum_{k\in M_\mathcal{P}} q_{ik})=0$. Thus, 
\begin{align}\nonumber
p_i=\frac{\sum_{k\in M_\mathcal{P}}q_{ik}}{\sum_{k\in M_\mathcal{P}}(q_{ik}-q^2_{ik})}\leq \frac{\sum_{k\in M_\mathcal{P}}q_{ik}}{\sum_{k\in M_\mathcal{P}}(q_{ik}-0.91q_{ik})}< 12.
\end{align}
If $\sum_{k\in M_\mathcal{P}} q_{ik}\ge c_i$, we can write,
\begin{align}\nonumber
1\!\leq\! c_i\!\leq \!\!\!\sum_{k\in M_\mathcal{P}}\!\!\!q_{ik}&\leq \!\!\!\sum_{k\in M_\mathcal{P}}\!\!\!\frac{\exp(\theta_{ik}-p_i)}{1+\exp(\theta_{ik}-p_i)}\!\leq\! \frac{m\exp(\theta-p_i)}{1+\exp(\theta-p_i)}.
\end{align}      
That shows that $p_i\leq \theta+\ln(m-1)$. Therefore, in either case, we have $p_i\leq \max\{12, \theta+\ln(m-1)\}=O(\ln m)$. Now we can upper-bound the revenue of the optimal pool $\mathcal{P}$ as 
\begin{align}\nonumber
\sum_{i\in N_\mathcal{P}}p_i\big(\sum_{k\in M_\mathcal{P}}q_{ik}\big)&\leq O(\ln m)\sum_{i\in N_\mathcal{P},k\in M_\mathcal{P}}q_{ik}\leq O(\ln m)\cdot \min\{|M_\mathcal{P}|, \sum_{i\in N_\mathcal{P}}c_i\},
\end{align}
where the second inequality is by Corollary \ref{corr:capacity} and the fact that the total demand is at most $|M_\mathcal{P}|$ (as each buyer $k\in M_\mathcal{P}$ contributes at most one unit to the overall demand). Finally, if we sum the above inequality over all optimal pools, the revenue obtained from optimal segmentation is at most 
\begin{align}\nonumber
\sum_{\mathcal{P}}O(\ln m) \min\{|M_\mathcal{P}|, \sum_{i\in N_\mathcal{P}}c_i\}\leq O(\ln m) \min\{\sum_{\mathcal{P}} |M_{\mathcal{P}}|, \sum_{\mathcal{P}}\sum_{i\in N_{\mathcal{P}}}c_i\}=O(\ln m) \min\{m, \sum_{i=1}^n c_i\}.
\end{align}
Thus, we only need to provide a pool partitioning whose revenue is within a constant factor of $\min\{m, \sum_{i=1}^n c_i\}$.

\begin{figure}[t]
\begin{center}
\vspace{-2cm}
\hspace{1cm}
\includegraphics[totalheight=.25\textheight,
width=.35\textwidth,viewport=0 0 850 850]{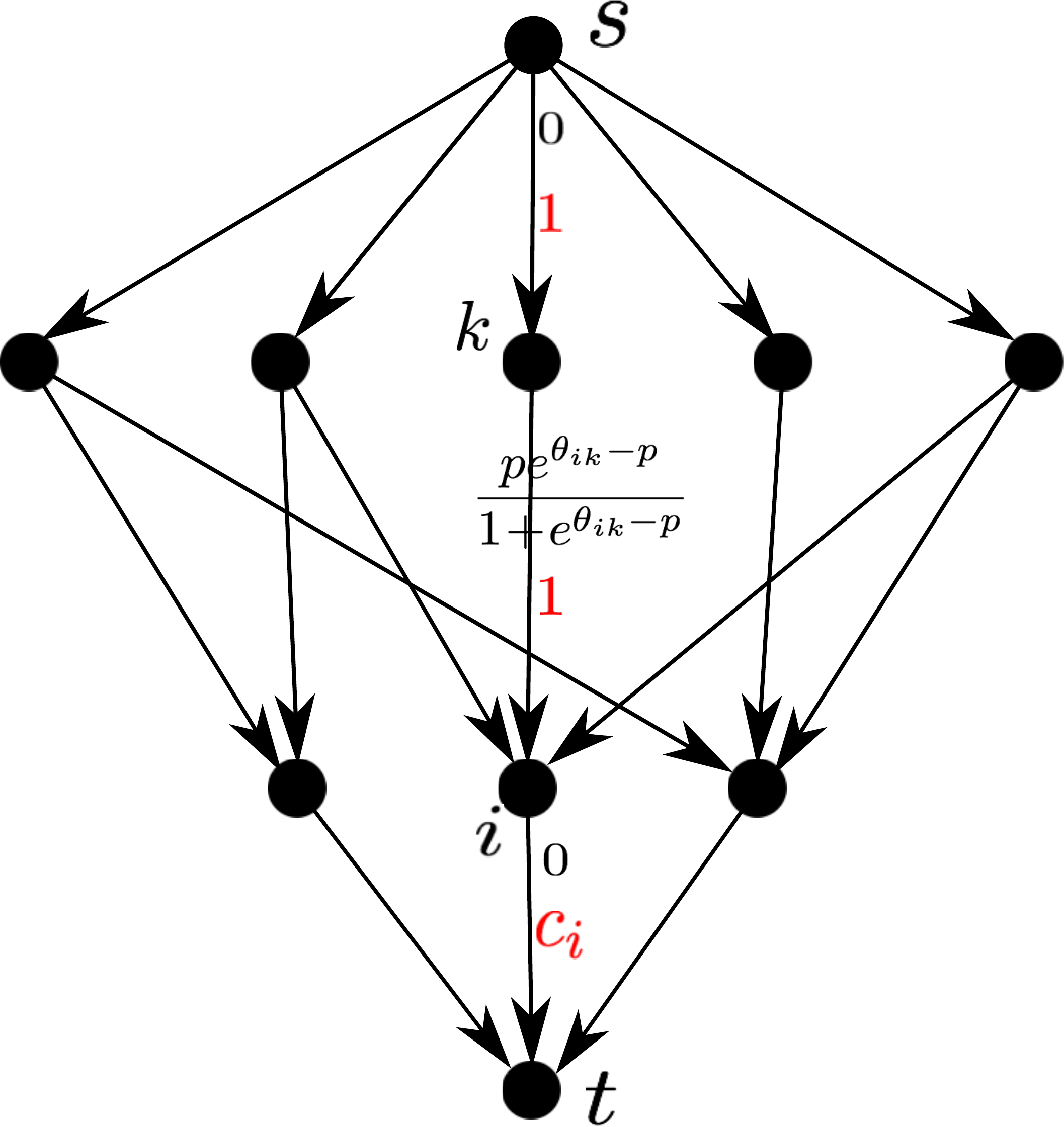} \hspace{0.4in}
\end{center}
\caption{An illustration of the network flow in the proof of Theorem \ref{thm:pool} with $m=5$ and $n=3$. The red and black values next to an edge represent the capacity and the weight of that edge.}\label{fig:flow}
\end{figure}

Consider a directed flow network obtained from the underlying bipartite graph of the generalized Bertrand game and two extra nodes, $s$ and $t$ (see Figure \ref{fig:flow}). Let $s$ be a source node connected to each buyer $k\in [m]$, where the capacity of edge $(s,k)$ is $1$ and its weight equals to $0$. Moreover, let $t$ be a sink node that is connected to each seller $i$ by a directed edge $(i,t)$ with capacity $c_i$ and weight $0$. For every other connected pair of a buyer and a seller $(k,i)$, we set the capacity of the directed edge $(k,i)$ to $1$ and its weight to $p\frac{\exp(\theta_{ik}-p)}{1+\exp(\theta_{ik}-p)}$, where, for the remainder of the proof, we set $p=1$. Let us consider the max-weight flow that sends $\min\{m,\sum_{i=1}^nc_i\}$ units of flow from $s$ to $t$. Since all the edge capacities are integral, the optimal flow is also integral, and assigns each buyer to at most one seller. Therefore, the set of edges that carry positive flow in the optimal flow will decompose the network into pools $\mathcal{P}_i, i\in N\subseteq [n]$, where a pool $\mathcal{P}_i$ comprises exactly one seller $i$ and possibly multiple buyers. In particular, because of the capacity constraints of the edges $(k,t)$, the number of buyers in $\mathcal{P}_i$ is at most $c_i$.

Now suppose that we segment the market into pools $\mathcal{P}_i, i\in N$ that are obtained from the above max-flow solution, and let each one equilibrate at a price $\hat{p}_i, i\in N$. Since $\hat{p}_i$ is the equilibrium price for seller $i$ in $\mathcal{P}_i$, the revenue obtained from $\mathcal{P}_i$ at equilibrium price $\hat{p}_i$ is no less than the revenue obtained from $\mathcal{P}_i$ at the unit price $p_i=1$. However, the revenue at the unit price $p_i=1$ is exactly the weight of the max-flow in $\mathcal{P}_i$. As the argument holds for each pool $\mathcal{P}_i, i\in N$, the overall revenue obtained by segmenting the market into pools $\mathcal{P}_i, i\in N$ is at least the weight of the overall max-flow. Finally, note that every edge $(k,i)$ in the flow network has weight $ \frac{\exp(\theta_{ik}-1)}{1+\exp(\theta_{ik}-1)}\ge \frac{1}{1+e}$. Thus, the weight of the max-flow is at least $\frac{\min\{m,\sum_{i=1}^n c_i\}}{1+e}$. 
\end{proof}

Finally, we note that in the proof of Theorem \ref{thm:pool}, we did not use the full power of the max-flow solution. In other words, we merely used the structure of the optimal flow to assign buyers to different pools while respecting capacity constraints. However, a distinguishing feature of such max-flow partitioning is that its solution incorporates the relative size of quality evaluations $\theta_{ik}$ into partitioning. Unfortunately, because of the highly nonlinear structure of the equilibrium demands and the lack of closed-form solutions for the equilibrium prices, it is not clear how to leverage that advantage to obtain a (possibly) constant-approximation algorithm for the optimal segmentation.

\section{Conclusions}\label{sec:conclusion}
In this paper, we considered an online assortment problem under a discriminatory control model wherein the platform may display only a subset of sellers to an arriving buyer. That situation induces competition among the sellers such that the sellers set their prices based on the Nash equilibrium of a single-buyer Bertrand game. We addressed the problem of finding a competitive online algorithm under inventory constraints with an unknown number of buyers and initial inventories. It is a challenging problem due to the coupling among the revenue of items, the offered assortments, and the inventory constraints. However, we showed that a simple hybrid algorithm achieves a constant competitive ratio and can be implemented in polynomial time $O(mn)$, where $n$ is the number of items and $m$ is the number of buyers. We also showed that the online assortment problem with heterogeneous buyers does not admit a constant competitive algorithm. To account for heterogeneous buyers, we then considered an offline setting in which different buyers can have different evaluations of items' qualities. We showed that under a mild consistency assumption, the offline generalized Bertrand game admits a pure Nash equilibrium, and we devised a simple $O(\ln m)$-approximation algorithm for its optimal segmentation.

This work opens several future directions for research.  For instance, one can consider an online assortment problem in which the information on sellers' inventories is public so that each seller can optimize its price as a function of others' inventory levels.  However,  such an extension could be problem-specific depending on what inventory pricing scheme one would adopt. Moreover,  it would be interesting to see whether the optimal segmentation of the generalized Bertrand game admits a constant factor approximation algorithm. Such an improvement requires a finer characterization of the Nash equilibrium prices in the generalized Bertrand game. Having said that, it may very well be that the approximation algorithm given here is close to optimal. If it is, it would be interesting to establish such a hardness result.

\section{Appendix I: Omitted Proofs}\label{sec:appx}

\smallskip
{\bf Proof of Lemma \ref{lemm:subtitute}}: First note that offering more items reduces the probability that no item will be purchased at the equilibrium, because if $q_0(S)$ denotes the no-purchase probability at the equilibrium, then $1-q_0(S)=\sum_{r\in S}V(q_0(S)e^{\theta_r-1})$, where $V(\cdot)$ is a strictly increasing function. Now if we offer a larger assortment $S\cup\{j\}$, we must have $q_0(S\cup\{j\})\leq q_0(S)$. Otherwise, offering $S\cup\{j\}$ can only increase the right side of the former equality while decreasing its left side, implying that $1-q_0(S\cup\{j\})<\sum_{r\in S\cup\{j\}}V(q_0(S\cup\{j\})e^{\theta_r-1})$. This contradiction shows that $q_0(S\cup\{j\})\leq q_0(S)$. Now, by monotonicity of $V(\cdot)$ and using \eqref{eq:p-q_0}, we have,
\begin{align}\nonumber
q_i(S)=V(q_0(S)e^{\theta_i-1})\ge V(q_0(S\cup\{j\})e^{\theta_i-1})=q_i(S\cup\{j\}).
\end{align}   
In other words, offering more items in the assortment reduces the market share for the existing ones. That also implies that the revenue derived from an item $i$ if it is offered in a larger set is less than when it is offered in a smaller set, as 
\begin{align}\nonumber
R_i(S)=\frac{q_i(S)}{1-q_i(S)}\ge \frac{q_i(S\cup\{j\})}{1-q_i(S\cup\{j\})}=R_i(S\cup\{j\}). 
\end{align}\hfill{$\blacksquare$}

\medskip
{\bf Proof of Lemma \ref{lemm-heavy-two}:} As $i$ is a heavy item, $q_i(\{i\})\ge \lambda$. Now let us first assume $i\in S$, meaning that $i$ is the heaviest item in $S$.

{\bf Case I}: If $q_i(S)\ge \lambda$, by $\sum_{j\in S\cup\{0\}}q_j(S)=1$, we have $\sum_{j\in S\setminus\{i\}}q_j(S)\leq 1-\lambda$. Thus $q_j(S)\leq 1-\lambda \ \forall j\in S\setminus\{i\}$, and we can write
\begin{align}\nonumber
R(S)&=\frac{q_i(S)}{1-q_i(S)}+\sum_{j\in S\setminus \{i\}}\frac{q_j(S)}{1-q_j(S)}\leq \frac{q_i(S)}{1-q_i(S)}+\frac{1-\lambda}{\lambda}\cr 
&\leq \big(1+(\frac{1-\lambda}{\lambda})^2\big)\frac{q_i(S)}{1-q_i(S)} \leq \big(1+(\frac{1-\lambda}{\lambda})^2\big)\frac{q_i(\{i\})}{1-q_i(\{i\})},
\end{align}
where the second inequality holds because $\frac{q_i(S)}{1-q_i(S)}\ge \frac{\lambda}{1-\lambda}$, and the last inequality holds because $q_i(S)\leq q_i(\{i\})$ by the substitutability property (Lemma \ref{lemm:subtitute}).

\noindent
{\bf Case II}: If $q_i(S)< \lambda$, since $i$ is the heaviest item in $S$, for any other $j\in S$ we must have $q_j(S)< \lambda$. As a result,
\begin{align}\label{eq:chain-light}
R(S)=\sum_{j\in S}\frac{q_j(S)}{1-q_j(S)}\leq \frac{1}{1-\lambda}\sum_{j\in S}q_j(S)\leq \frac{1}{1-\lambda}\leq \frac{1}{\lambda}\frac{q_i(\{i\})}{1-q_i(\{i\})},
\end{align} 
where the last inequality holds because $q_i(\{i\})\ge \lambda$.

Finally, if $i\notin S$, then either $S$ does not contain any heavy item, in which case $q_j(S)<\lambda, \forall j\in S$, and the same chain of inequalities in \eqref{eq:chain-light} holds, or $S$ contains at least one heavy item. In the latter case, let $k\in S$ be the heaviest item in $S$. Now, using the proof of Case I, we have $R(S)\leq \big(1+(\frac{1-\lambda}{\lambda})^2\big)\frac{q_k(\{k\})}{1-q_k(\{k\})}$. Since by the assumption, $i$ is heavier than $k$, $q_k(\{k\})\leq q_i(\{i\})$, implying $R(S)\leq \big(1+(\frac{1-\lambda}{\lambda})^2\big)\frac{q_i(\{i\})}{1-q_i(\{i\})}$. Therefore, if we define $f(\lambda)=\max\{1+(\frac{1-\lambda}{\lambda})^2, \frac{1}{\lambda}\}$, both of the above cases hold and we have $R(S)\leq f(\lambda)\frac{q_i(\{i\})}{1-q_i(\{i\})}$.

\medskip
{\bf Proof of Lemma \ref{lemm:basic-Big}:} Consider a virtual online assortment problem with initial inventory $\boldsymbol{c}$ and purchase probabilities $\boldsymbol{q}(S)$ as in the original online assortment problem, except that each item $i$ has a \emph{fixed} constant revenue $r_i$. In other words, the expected revenue obtained by including $i$ in an assortment $S$ equals to $R_i(S)=r_iq_i(S)$. Using an argument similar to that in Section \ref{eq:preliminary}, we find that the optimal clairvoyant online revenue for the virtual problem is upper-bounded by the optimal value of the following LP: 
\begin{align}\label{eq:basic}
\mbox{OPT}(\boldsymbol{r})=\max &\ \ \sum_{t=1}^m\sum_{S} \Big(\sum_i r_iq_i(S)\Big)y^t(S)\cr 
\mbox{s.t.}  &\ \ \sum_{t=1}^m\sum_{S} \boldsymbol{q}(S)y^t(S)\leq \boldsymbol{c},\cr 
& \ \ \boldsymbol{y}^t\in \Delta_{2^n}, \forall t=1,\ldots,m,
\end{align}  
where $\boldsymbol{r}=(r_1,\ldots,r_n)$. It has been shown in \cite[Example 1]{golrezaei2014real} that a greedy online algorithm that at time $t$ offers the maximizing assortment 
\begin{align}\label{eq:greedy-policy}
S_t=\argmax_{S\subseteq \mathcal{A}_t}\sum_{i\in S} r_iq_i(S),
\end{align}
is $\frac{1}{2}$-competitive with respect to OPT$(\boldsymbol{r})$, where $\mathcal{A}_t$ denotes the set of available items at time $t$. Since both the original and virtual assortment problems, as well as their offline LP benchmarks \eqref{eq:LP-Game} and \eqref{eq:basic}, share the same inventory constraints, any feasible online algorithm for one is also feasible for the other one. The only difference is in their expected objective revenues that are given by $\sum_{t=1}^{m}\mathbb{E}\big[\sum_{i\in S_t}\frac{q_i(S_t)}{1-q_i(S_t)}\big]$ and $\sum_{t=1}^{m}\mathbb{E}\big[\sum_{i\in S_t}r_iq_i(S_t)\big]$, respectively. Now, let $\boldsymbol{r}=\boldsymbol{1}$, and denote the expected revenue of the greedy algorithm on the original and virtual problems by $\mbox{Rev}^{\rm o}_{\rm g}$ and $\mbox{Rev}^{\rm v}_{\rm g}(\boldsymbol{1})$, respectively. As $\frac{q_i(S)}{1-q_i(S)}\ge q_i(S), \forall i, S$, a simple coupling shows that
\begin{align}\nonumber
\mbox{Rev}^{\rm o}_{\rm g}\ge \mbox{Rev}^{\rm v}_{\rm g}(\boldsymbol{1})\ge \frac{1}{2} \mbox{OPT}(\boldsymbol{1}).
\end{align}
Finally, when $\boldsymbol{r}=\boldsymbol{1}$, at each time $t$, the greedy rule in \eqref{eq:greedy-policy} offers the assortment
\begin{align}\nonumber
S_t=\argmax_{S\subseteq \mathcal{A}_t}\sum_{i\in S} q_i(S)=\argmax_{S\subseteq \mathcal{A}_t}\Big(1-q_0(S)\Big),
\end{align}  
where by \eqref{eq:q_0}, the right side is maximized if $S_t=\mathcal{A}_t$.\hfill{$\blacksquare$}

\medskip
{\bf Maximizing $g(\lambda)$ in the Proof of Theorem \ref{lemm:basic-Big}:} In order to compute 
\begin{align}\nonumber
g(\lambda)=\min_{y\ge 0}\ \max_{0\leq x\leq \min\{y,1\}}\Big\{\max\{0, 1-\frac{y}{\lambda}\}\big(\frac{1-\lambda}{2}-\frac{x}{f(\lambda)+x}\big)+\frac{(\min\{\lambda,y\}-x)x}{(f(\lambda)+x)(1-x)}\Big\},
\end{align}
we consider two cases: If $y> \lambda$ (i.e., if the number of heavy items $c$ is more than the number of buyers), then $\max\{0, 1-\frac{y}{\lambda}\}=0$, and we get $g(\lambda)=\max\limits_{0\leq x\leq \lambda}\Big\{\frac{(\lambda-x)x}{(f(\lambda)+x)(1-x)}\Big\}$. This maximization can be solved analytically to get
\begin{align}\nonumber
g(\lambda)=\Big(\frac{\sqrt{(1-\lambda)(\lambda+f(\lambda))}-\sqrt{f(\lambda)}}{\sqrt{f(\lambda)(\lambda+f(\lambda))}-\sqrt{1-\lambda}}\Big)^2,
\end{align}  
where $f(\lambda)=\max\{1+(\frac{1-\lambda}{\lambda})^2, \frac{1}{\lambda}\}$. (See the blue curve in Figure \ref{fig:g}.) Otherwise, if $y\leq \lambda$, we have 
\begin{align}\nonumber
g(\lambda)=\min\limits_{0\leq y \leq \lambda}\ \max\limits_{0\leq x\leq y}\Big\{(1-\frac{y}{\lambda})(\frac{1-\lambda}{2}-\frac{x}{f(\lambda)+x})+\frac{(y-x)x}{(f(\lambda)+x)(1-x)}\Big\}.
\end{align}
This function is computed numerically and is depicted in Figure \ref{fig:g} using a red curve. The lower envelope of these two functions determines a lower bound for the competitive ratio of Algorithm \ref{alg-main} for different values of $\lambda$. Finally, by maximizing that lower envelope over $\lambda$, one can see that both curves achieve a competitive ratio of at least $0.057$. \hfill{$\blacksquare$}  
\begin{figure}[t]
\vspace{-1.5cm}
\begin{center}
\includegraphics[totalheight=.25\textheight,
width=.35\textwidth,viewport=0 0 500 500]{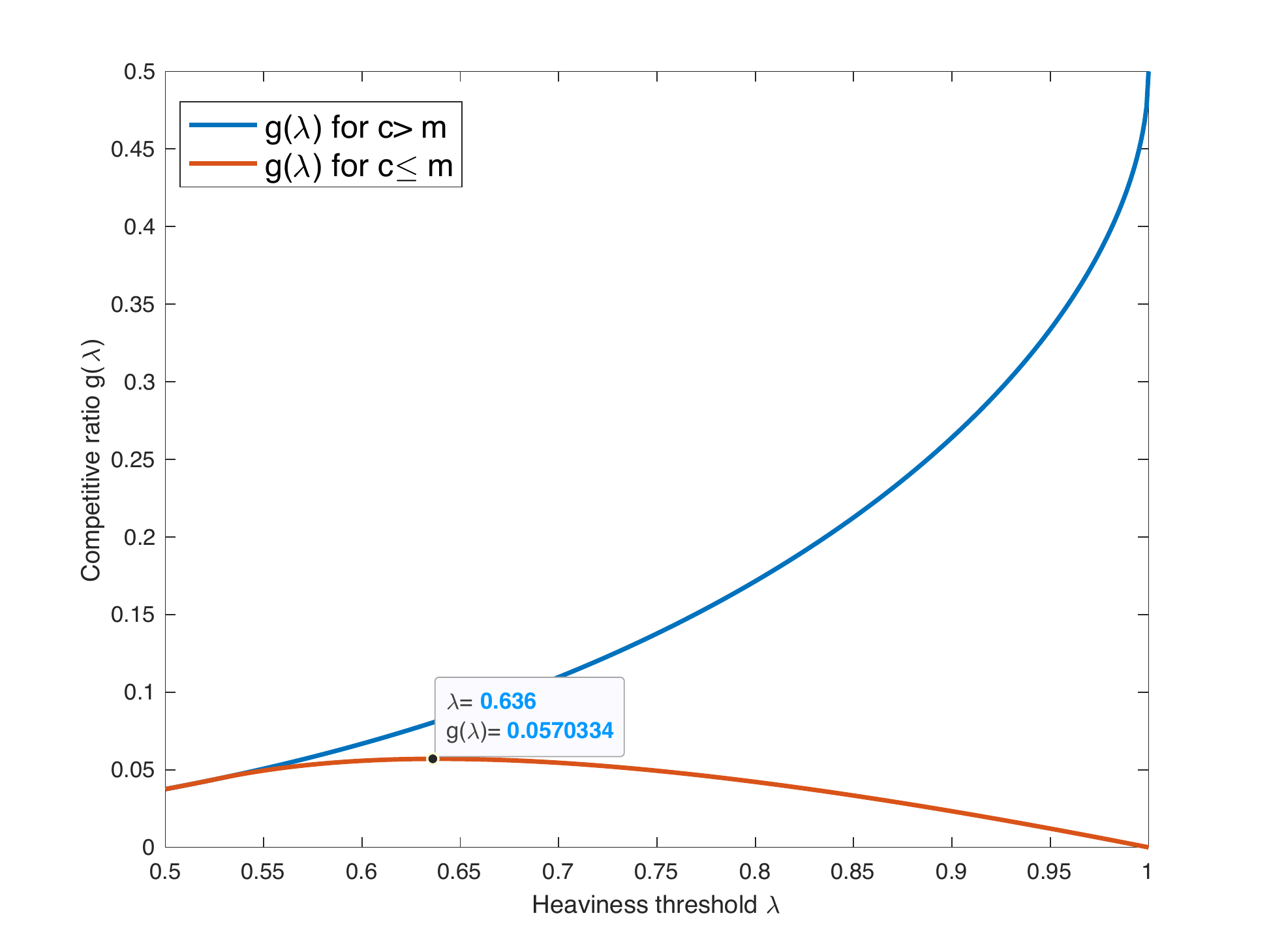} \hspace{0.4in}
\end{center}\vspace{-0.3cm}
\caption{Competitive ratio function $g(\lambda)$ for two different cases of $c> m$ and $c\leq m$.}\label{fig:g}
\end{figure}
 
\section{Appendix II: Numerical Experiments}\label{sec:sim}
In this section, we provide the results of some numerical experiments to demonstrate the efficiency of the hybrid Algorithm \ref{alg-main} beyond the theoretical guarantee given in Theorem \ref{them:online}. In our simulations, we consider a set of $n=10$ items with qualities $\theta_{10}=-2, \theta_9=-1.5, \theta_8=-1, \theta_7=-0.5, \theta_6=0.5, \theta_5=1, \theta_4=1.5, \theta_3=2, \theta_2=2.5, \theta_1=3. $, and inventory levels $c_i=15, \forall i\in[10]$. In each figure, we evaluate the competitive ratio of the hybrid Algorithm \ref{alg-main} by changing only one parameter while keeping all other parameters fixed. More precisely, we consider the competitive ratio $\alpha=\frac{\mathbb{E}[R({\rm Alg})]}{{\rm OPT}}$ under three different scenarios. In the left side of Figure \ref{fig-1}, we have changed the number of buyers from $m=100$ to $m=500$ while setting the heaviness threshold to $\lambda=0.5$. As can be seen, the competitive ratio increases as the number of items increases, and it approaches $1$ for a large number of buyers. The reason is that for a large $m$, the Algorithm \ref{alg-main} has enough time to sell each item at its maximum price by offering it alone. On the right side of Figure \ref{fig-1}, we again set $\lambda=0.5$ and increase the inventory of the first three items from $1$ to $30$ (while keeping all others' inventories fixed). Finally, in Figure \ref{fig-2}, we have illustrated the effect of change of heaviness threshold $\lambda$ in the competitive ratio. As $\lambda$ changes from $\lambda=0.5$ to $\lambda=0.9$, the competitive ratio changes between $0.2$ to $0.37$, with its maximum value achieved for the threshold $\lambda=0.63$. While that optimal threshold matches our theoretical analysis, however, it is worth noting that in this numerical experiment, the competitive ratio for $\lambda=0.63$ is at least $\alpha\ge 0.37$, which is pretty good and higher than the theoretical guarantee $\alpha\ge 0.057$. In summary, the competitive ratio of the hybrid Algorithm \ref{alg-main}, in general, has a complicated nonlinear dependence on each of the problem parameters. However, as can be seen from all the figures, $\alpha\ge 0.2$ for $\lambda=0.5$, and $\alpha\ge 0.37$ for $\lambda=0.63$. In particular, the competitive ratio can be close to $1$ for a certain range of parameters. 

\begin{figure}[!tbp]
  \centering
  \hspace{-0.3cm}\begin{minipage}[t]{0.38\textwidth}
    \includegraphics[width=\textwidth]{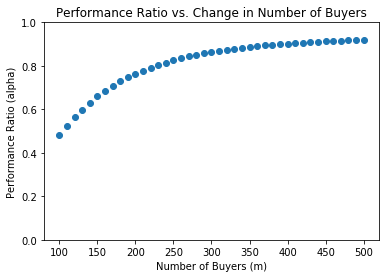}
  \end{minipage}
  \hspace{2.1cm}
  \begin{minipage}[t]{0.38\textwidth}
    \includegraphics[width=\textwidth]{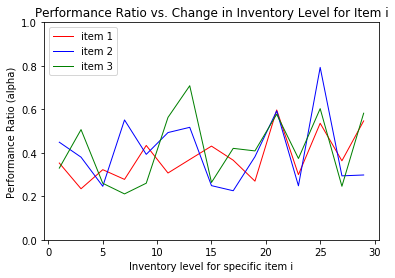}
  \end{minipage} \caption{\footnotesize{Competitive ratio of the hybrid Algorithm \ref{alg-main} with $\lambda=0.5$ v.s. change in the number of buyers (left figure), and change in the inventory levels (right figure).}}\label{fig-1}
\end{figure}

\begin{figure}[!tbp]
  \centering
  \begin{minipage}[t]{0.39\textwidth}
    \includegraphics[width=\textwidth]{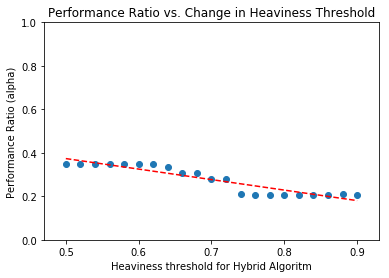}
    \vspace{-0.3cm}
    \caption{\footnotesize{Competitive ratio v.s. change in the heaviness threshold of the hybrid algorithm.}}\label{fig-2}
  \end{minipage}
  \hspace{2cm}
  \begin{minipage}[t]{0.41\textwidth}
    \includegraphics[width=\textwidth]{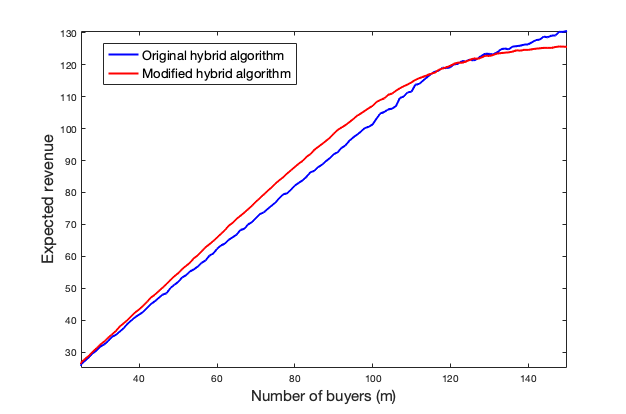}
    \vspace{-0.4cm}
    \caption{\footnotesize{Expected revenue v.s. number of buyers for the hybrid and modified hybrid algorithms.}}\label{fig-3}
  \end{minipage} 
  \vspace{-0.5cm}
\end{figure}


\subsection{Improvement Using Information on Sellers' Inventories} In our work, we assumed that the platform does not have access to the information on sellers' inventories. That not only respects the sellers' privacy but also having access to such information may not be feasible in some inventory management problems \cite{wang2019inventory} (e.g., due to inaccuracies in inventory recording, misplaced products, market fluctuation, etc.).  However, the online hybrid algorithm achieves a constant competitive ratio even under this restricted setting, and so our results are robust if the platform has more information about the system. In particular, the platform can only benefit from that extra information to improve its performance. Here, we illustrate numerically how extra information on sellers' inventories can help the platform improve its performance. To that aim, we consider a modified version of the online hybrid algorithm, which also incorporates the information on the sellers' inventory states into its assortment recommendations.  The modified hybrid algorithm dynamically scales the heaviness of the items according to the remaining inventories and then offers the assortments based on the relative heaviness of the items.  More precisely, let $\Psi:[0,1]\to [0,1]$ be a non-decreasing weight function and let us denote the inventory level of seller $i$ at time $t$ by $I_i(t)$. Then we can define the relative heaviness of item $i$ at time $t$ by $\Psi(\frac{I_i^t}{c_i})q_i(\{i\})$, where we recall that $c_i$ is the initial inventory of item $i$, and $q_i(\{i\})$ is the equilibrium demand for item $i$ if it is offered alone in an assortment. Now, if an item $i$ has low inventory,  the weighting factor $\Psi(\frac{I_i^t}{c_i})$ would be small so that the modified hybrid algorithm gives less priority to low inventory items in order to save them for future sales. Consequently, the algorithm offers the items with larger inventory more aggressively. 

We have simulated the performance of the modified hybrid algorithm with exponential weight function $\Psi(x)=\frac{e}{e-1}(1-e^{-x})$ in Figure \ref{fig-3}.\footnote{This choice of weight function is motivated by the inventory balancing algorithm given in \cite{golrezaei2014real}.} We consider $n=10$ items with qualities $\boldsymbol{\theta}=(2.1, 2,2,2,2,0.5,-0.5,-1,-1.5,-2)$. We also set the initial inventories to $\boldsymbol{c}=(20,20,20,20,20,5,5,5,5,5)$ and the heaviness threshold to $\lambda=\frac{1}{2}$.  The expected revenue of the hybrid algorithm (blue curve) and its modified version (red curve) as the number of buyers changes from $m=25$ to $m=150$ are shown in Figure \ref{fig-3}. As can be seen, for the mid-range number of buyers, the modified algorithm with inventory balancing weight function indeed outperforms the original hybrid algorithm. However, for many buyers, the platform has enough time to sell the items individually at their highest price, and hence the original hybrid algorithm that is oblivious to remaining inventories performs better.

\bibliographystyle{IEEEtran}
\bibliography{thesisrefs}
\end{document}